\documentclass[11pt,letterpaper]{article}

\usepackage[T1]{fontenc}
\usepackage{lmodern}
\usepackage{microtype}
\usepackage[letterpaper,margin=1in]{geometry}
\usepackage{amsmath,amssymb,amsthm,mathtools}
\usepackage{aliascnt}
\usepackage{enumitem}
\usepackage{tabularx,environ,caption}
\usepackage[numbers,sort&compress]{natbib}
\usepackage{xurl}
\usepackage[colorlinks=true,linkcolor=blue,citecolor=blue,urlcolor=blue]{hyperref}
\usepackage[nameinlink,noabbrev,capitalise]{cleveref}

\usepackage{todonotes}

\makeatletter
\newcolumntype{\expand}{}
\long\@namedef{NC@rewrite@\string\expand}{\expandafter\NC@find}

\NewEnviron{problem}[2][]{%
  \def\problem@arg{#1}%
  \def\problem@framed{framed}%
  \def\problem@lined{lined}%
  \def\problem@doublelined{doublelined}%
  \ifx\problem@arg\@empty
    \def\problem@hline{}%
  \else
    \ifx\problem@arg\problem@doublelined
      \def\problem@hline{\hline\hline}%
    \else
      \def\problem@hline{\hline}%
    \fi
  \fi
  \ifx\problem@arg\problem@framed
    \def\problem@tablelayout{|>{\itshape}lX|c}%
    \def\problem@title{\multicolumn{2}{|l|}{%
        \raisebox{-\fboxsep}{\textsc{#2}}%
      }}%
  \else
    \def\problem@tablelayout{>{\itshape}lXc}%
    \def\problem@title{\multicolumn{2}{l}{%
        \raisebox{-\fboxsep}{\textsc{\large #2}}%
      }}%
  \fi
  \par\noindent
  \renewcommand{\arraystretch}{1.0}%
  \begin{center}
  \begin{tabularx}{0.9\columnwidth}{\expand\problem@tablelayout}%
    \problem@hline
    \problem@title\\[2\fboxsep]
    \BODY\\\problem@hline
  \end{tabularx}%
  \end{center}
  \par
}
\makeatother

\setlist[enumerate]{topsep=.4em,itemsep=.1em,parsep=0pt,partopsep=0pt}

\newtheorem{theorem}{Theorem}[section]
\newaliascnt{lemma}{theorem}
\newtheorem{lemma}[lemma]{Lemma}
\aliascntresetthe{lemma}
\newaliascnt{proposition}{theorem}
\newtheorem{proposition}[proposition]{Proposition}
\aliascntresetthe{proposition}
\newaliascnt{corollary}{theorem}
\newtheorem{corollary}[corollary]{Corollary}
\aliascntresetthe{corollary}
\theoremstyle{definition}
\newaliascnt{definition}{theorem}

\aliascntresetthe{definition}
\newaliascnt{example}{theorem}

\aliascntresetthe{example}
\theoremstyle{remark}
\newaliascnt{remark}{theorem}
\newtheorem{remark}[remark]{Remark}
\aliascntresetthe{remark}

\usepackage{thm-restate}

\crefname{theorem}{theorem}{theorems}
\crefname{lemma}{lemma}{lemmas}
\crefname{proposition}{proposition}{propositions}
\crefname{corollary}{corollary}{corollaries}
\crefname{definition}{definition}{definitions}
\crefname{example}{example}{examples}
\crefname{remark}{remark}{remarks}

\newcommand{\R}{\mathbb R}
\newcommand{\cH}{\mathcal H}
\newcommand{\cA}{\mathcal A}
\newcommand{\Hom}{\mathsf{Hom}}
\newcommand{\pHom}{p\text{-}\mathsf{Hom}}
\newcommand{\ghw}{\mathsf{ghw}}
\newcommand{\fhw}{\mathsf{fhw}}
\newcommand{\adw}{\mathsf{adw}}
\newcommand{\subw}{\mathsf{subw}}
\newcommand{\sharpsubw}{\mathsf{\#subw}}
\newcommand{\fsep}{\mathsf{fsep}}
\newcommand{\FPT}{\ensuremath{\mathsf{FPT}}}
\newcommand{\PTIME}{\mathsf{PTIME}}
\newcommand{\NP}{\mathsf{NP}}
\newcommand{\QP}{\mathsf{QP}}
\newcommand{\rank}{\operatorname{rank}}
\newcommand{\supp}{\operatorname{supp}}

\newcommand{\emb}{\operatorname{emb}}
\newcommand{\wone}{\ensuremath{\mathsf{W}[1]}}
\newif\ifaidisclosure
\aidisclosuretrue

\title{\FPT{} iff $\PTIME$ for Homomorphism Problems\\
on Sparse-Incidence and Bounded-Independence Patterns}
\author{Matthias Lanzinger\\
Institute of Logic and Computation, TU Wien}
\date{}

\providecommand{\adw}{\operatorname{adw}}
\providecommand{\subw}{\operatorname{subw}}
\providecommand{\fhw}{\operatorname{fhw}}
\providecommand{\FPT}{\ensuremath{\mathsf{FPT}}}
\providecommand{\PTIME}{\mathsf{PTIME}}

\makeatletter
\@ifundefined{theorem}{\newtheorem{theorem}{Theorem}[section]}{}
\@ifundefined{corollary}{\newtheorem{corollary}[theorem]{Corollary}}{}
\@ifundefined{lemma}{\newtheorem{lemma}[theorem]{Lemma}}{}
\@ifundefined{proposition}{\newtheorem{proposition}[theorem]{Proposition}}{}
\@ifundefined{remark}{}{}
\makeatother

\begin{document}

\pagenumbering{arabic}
\maketitle

\begin{abstract}
\normalsize
We show that fixed-parameter tractability and polynomial-time solvability coincide for homomorphism problems specified by classes of pattern hypergraphs of bounded incidence degeneracy or bounded primal independence number, assuming $\FPT \ne \wone$. In both cases, tractability is characterised by bounded fractional hypertree width. This generalises a result of Grohe [JACM 2007] that established such an \FPT{}--PTIME collapse for bounded-arity structures.
On a technical level, we build on the fractional balanced-separator framework and rounding theorem of Korchemna, Lokshtanov, Saurabh, Surianarayanan, and Xue [FOCS 2024] to prove near-linear bounds on fractional hypertree width ($\fhw$) in terms of adaptive width ($\adw$). Namely, for every
hypergraph $H$ with $\adw(H)\ge2$,
\[
\fhw(H)=O\!\left(\mu(H)\,\adw(H)\log\adw(H)\right), \quad
\fhw(H)=O\!\left(\log\alpha(H)\,\adw(H)\log\adw(H)\right)
\]
where $\mu(H)$ is the degeneracy of the incidence graph of $H$ and $\alpha(H)$ the independence number of the primal graph of $H$.
As a consequence, we obtain a corresponding \FPT{}--PTIME collapse for exact homomorphism decision and counting on every class in which either incidence degeneracy or independence number is bounded. This includes a number of natural classes such as pattern hypergraphs of bounded degree, and patterns whose incidence graphs exclude a fixed topological minor.
Additionally, we observe that for every recursively enumerable class of pattern hypergraphs, fixed-parameter tractability of the parameterised homomorphism problem implies quasipolynomial-time solvability of the corresponding unparameterised problem, assuming $\FPT \ne \wone$. This suggests a new structural route towards identifying NP-intermediate problems.

Our results rely on the characterisation of fixed-parameter tractable decision of homomorphism problems in terms of submodular width of Marx [JACM 2013], conditional on the Exponential Time Hypothesis (ETH).
In light of the very recent potential refutation of the ETH, we show that the characterisation of Marx holds also under the assumption that $\FPT \ne \wone$.
\end{abstract}
 \clearpage

\tableofcontents
\clearpage

\section{Introduction}
\label{sec:introduction}

We study the homomorphism problem between relational structures. For relational structures $A,B$ over the same vocabulary, a homomorphism from $A$ to $B$ is a map from the domain of $A$ to the domain of $B$ that preserves every relation. We write $A \to B$ when such a map exists.
This is known to generalise a wide range of common algorithmic problems. For instance, on graph structures, if $B$ is fixed to be the $k$-clique $K_k$, then deciding $A \to K_k$ is equivalent to $k$-colourability of $A$. On the other hand, fixing $A=K_k$, we have that $K_k \to B$ is equivalent to the $k$-clique problem, and $K_k \to \bar B$ is the $k$-independent set problem.

In general, the problem is NP-complete. As these examples already suggest, it is of major interest to understand under which restrictions on the left-hand side (the \emph{patterns}), or on the right-hand side (the \emph{hosts} or \emph{templates}), the problem is tractable. To simplify further discussion, we refer to the setting in which the templates belong to a fixed class as the constraint satisfaction problem (CSP).
The complexity of restricting either side of the homomorphism problem has been studied extensively. For CSP this has culminated in a dichotomy between polynomial-time solvability and NP-completeness for each fixed finite template~\cite{Bulatov2017,Zhuk2020}.
When the left-hand side is restricted to a class of patterns, the picture is more complicated. Bodirsky and Grohe~\cite{BodirskyGrohe2008} showed that no dichotomy between polynomial-time solvability and NP-completeness holds for arbitrary pattern classes. In this paper we study restrictions to the left-hand sides. Formally, for a class of hypergraphs $\cH$ and $H_A$ denoting the hypergraph structure of $A$:
\begin{problem}{$\Hom(\cH)$}
Input & Finite relational structures $A$ and $B$
with $H_A\in\cH$.\\
Question & Does $A\to B$?
\end{problem}

We write $\pHom(\cH)$ for the parameterisation of $\Hom(\cH)$ by the
encoding length $\lVert H_A\rVert$.\footnote{We encode a hypergraph by listing its edges, each as a list of vertices.} In the database literature and in logic, $\Hom(\cH)$ corresponds to
the Boolean conjunctive query problem, with $A$ being the query and $B$ the
database on which it is evaluated~\cite{ChandraMerlin77}.

While there is no $\PTIME$ vs. $\NP$-complete dichotomy for $\Hom(\cH)$, there are still results on the limits of tractability. In a landmark result Grohe~\cite{DBLP:journals/jacm/Grohe07} proved that, for every
recursively enumerable hypergraph class $\cH$ of \emph{bounded rank}, the problem is fixed-parameter tractable exactly when it is feasible in polynomial time, assuming
$\FPT\ne\wone$. Moreover, these conditions hold precisely when $\cH$ has bounded treewidth, i.e., there is some constant upper bound on the treewidth of every hypergraph in $\cH$.
This leaves open the case of unbounded arity, where treewidth no longer captures tractability. For instance, a structure consisting of a single tuple with $r$ distinct entries has treewidth $r-1$, although its homomorphism problem is easy.

This motivated more ``hypergraph-native'' variants of treewidth, beginning
with hypertree width and generalised hypertree
width ($\ghw$)~\cite{DBLP:journals/jcss/GottlobLS02}, and then fractional
hypertree width ($\fhw$)~\cite{fhw}.
These notions are still based on tree decompositions, but introduce more complex cost functions for bags (in contrast to cardinality of the bag for treewidth). In particular, fractional hypertree width at most $k$ means that there is a tree decomposition such that every bag has fractional edge-cover number at most $k$ (formal definitions follow in \Cref{sec:preliminaries}). Analogously, generalised hypertree width ($\ghw$) uses integral edge covers instead.
Bounded fractional hypertree width
remains the most general known sufficient condition on the pattern
side for $\Hom(\cH)\in\PTIME$~\cite{fhw}.

The tractability boundary extends further in the parameterised setting.
Marx introduced \emph{adaptive width} ($\adw$) and \emph{submodular width} ($\subw$)~\cite{adw,Marx13}, which characterise precisely when $\pHom$ is fixed-parameter tractable.
These widths add a level of quantification to the problem. Where treewidth and the hypertree widths have a single decomposition as a witness, with a small cardinality or cover in every bag, adaptive and submodular width require a low cost decomposition independently \emph{for each} cost function in certain classes of functions. This additional quantification over admissible functions makes connecting the two kinds of width notions technically challenging.

The widths discussed so far form a well-known hierarchy~\cite{adw,Marx13,fhw}
\[
 \adw(H)\le\subw(H)\le\fhw(H)\le\ghw(H).
\]
Moreover, there are
classes of bounded adaptive width and unbounded fractional hypertree
width~\cite{adw}, so a reverse bound requires additional structure.
But the reverse direction is much less understood. For hypergraphs of maximum vertex degree
two, where each element occurs in at most two tuples,
Lanzinger~\cite{LanzingerPODS22} showed (conditional on complexity assumptions) that bounded fractional
hypertree width and bounded submodular width coincide for structures with degree 2, i.e., every domain element occurs in at most 2 tuples. This conclusion follows through the complexity characterisations and gives no explicit quantitative bound between the widths. More recently, \cite{Lanzinger2026Cuts} gives a first direct bound in the form $\ghw(H) = O(\subw(H)\log \subw(H))$
under two newly introduced structural conditions \emph{bounded edge excess} and \emph{private intersections}. However, both conditions are rather restrictive. For a detailed discussion of the connections to our main results, see~\Cref{sec:related-work}. Even apart from comparisons between width measures, analysing submodular width is technically challenging already on graphs. Bringmann and Gorbachev~\cite{BringmannGorbachev2025} classify the graph patterns admitting subquadratic algorithms for coloured subgraph listing, enumeration, and minimum-weight subgraph problems, with matching conditional lower bounds. They observe that analysing submodular width for these patterns would be essentially as involved as designing the algorithms directly.

Our main theorem gives a bound that closely couples $\fhw$ and $\adw$, connecting these two different branches of hypergraph width measures and their accompanying theories. The two are linked via either the degeneracy $\mu(H)$ of the incidence graph of $H$, or the logarithm of the independence number $\alpha(H)$ of the primal graph.
\begin{restatable}{theorem}{mainwidththeorem}
\label{thm:intro-width}
There is a universal constant $C>0$ such that every finite hypergraph $H$
with $\adw(H)\ge2$ satisfies
\[
 \fhw(H)\le C\mu(H)\adw(H)\log\adw(H) \qquad \fhw(H)\le C\max\{1,\log\alpha(H)\}\, \adw(H)\log\adw(H).
\]
When $\adw(H)<2$, $\fhw(H)\le C\mu(H)$ and $\fhw(H)\le C\max\{1,\log\alpha(H)\}$.
\end{restatable}
On a technical level this result builds on the hyperedge-weighted fractional balanced-separator relaxation and rounding theorem of~\cite{KorchemnaEtAl2024}. Originally these tools were used to efficiently compute approximations of $\fhw$.
We instead observe that the space of these fractional balanced separators is compact and convex. This in turn allows for a minimax argument that ultimately lets us bound fractional balanced separators in terms of adaptive width. From this we then show constructively that a tree decomposition adhering to the bound in \Cref{thm:intro-width} always exists.

\Cref{thm:intro-width} supplies the missing quantitative link between two
substantial bodies of work on the two sides of this comparison.
As a result, the theorem immediately reveals a variety of far-reaching consequences for $\Hom(\cH)$ and $\pHom(\cH)$ that we summarise here.

\paragraph{1) $\FPT$ iff $\PTIME$} Together with the width hierarchy, our main theorem shows that bounded fractional hypertree width, bounded submodular width, and bounded adaptive width are equivalent on every class of bounded incidence degeneracy or bounded primal independence.
For recursively enumerable hypergraph classes, bounded submodular width characterises fixed-parameter tractability under $\FPT\ne \wone$ (see the discussion of the Exponential Time Hypothesis below), while bounded fractional hypertree width guarantees polynomial-time solvability. Thus the two tractability boundaries coincide on these classes.

\begin{theorem}
\label{thm:intro-collapse}
Let $\cH$ be a recursively enumerable hypergraph class of bounded
incidence degeneracy or primal independence. Then, assuming $\FPT \ne \wone$,
\[
 \pHom(\cH)\in\FPT
 \quad\Longleftrightarrow\quad
 \Hom(\cH)\in\PTIME.
\]
These conditions hold exactly when $\cH$ has bounded fractional hypertree width.
\end{theorem}
Bounded incidence degeneracy includes hypergraphs of bounded rank or bounded vertex degree, as well as hypergraphs whose incidence graphs are planar or exclude a fixed topological minor. We also extend the characterisation to classes of relational structures by considering the widths of their cores, thereby generalising the theorem of Grohe mentioned above (see also \Cref{thm:semantic-incidence-collapse}).

\paragraph{2) $\FPT$ implies quasipolynomial time}
Since $\alpha(H)\le |V(H)|$, the independence bound turns a
constant adaptive-width bound into an $O(\log|V(H)|)$ fractional
hypertree-width bound. This gives the following consequence.
\begin{theorem}
\label{thm:intro-hypergraph-fpt-to-qp}
Let $\cH$ be a recursively enumerable hypergraph class.
Assuming $\FPT\ne \wone$, if $\pHom(\cH)\in\FPT$, then
$\Hom(\cH)$ can be solved in time
\(
 N^{(\log N)^{1+o(1)}},
\)
where $N=\lVert A\rVert+\lVert B\rVert$ is the input size.
\end{theorem}

Thus the existence of an algorithm with running time $f(\lVert H_A\rVert)N^{O(1)}$ also always implies an $N^{(\log N)^{1+o(1)}}$ time algorithm (note that the difference between $\lVert H_A\rVert$ and $\lVert A \rVert$ is immaterial).
We consider this rather surprising and to the best of our knowledge no such result was suspected.

The detailed runtime bound we obtain in the context of  \Cref{thm:intro-hypergraph-fpt-to-qp} is $N^{O((\log N)(\log\log \lVert H_A \rVert)^2)}$. While \Cref{thm:intro-width} implies that bounded submodular width implies $\fhw(H_A) \in O(\log \lVert H_A \rVert)$, we also need to find an appropriate decomposition to use this width. However, even for fixed $k$, deciding $\fhw(H)\le k$ is $\NP$-hard~\cite{DBLP:journals/jacm/GottlobLPR21}. We therefore need to rely on approximate decompositions for the pattern hypergraph. The best known approximation factor due to~\cite{KorchemnaEtAl2024} introduces the additional $(\log\log \lVert H_A \rVert)^2$ in the exponent.

Finding natural NP-intermediate problems beyond Ladner's diagonalisation remains a longstanding challenge (see, e.g., the discussion in Section~3.3 of~\cite{AroraBarak2009}). Another source of suspected intermediate complexity is problems involving objects of logarithmic size, such as logarithmic Clique and the smallest missing induced subgraph studied by Eppstein, Lincoln, and Vassilevska Williams~\cite{EppsteinLincolnWilliams2023}. Our result suggests a different route through structural restrictions on pattern hypergraphs. It would suffice to identify a necessary condition for polynomial-time solvability of $\Hom(\cH)$ that fails on some polynomial-time recognisable hypergraph class $\cH$ of bounded submodular width. The proof of \Cref{thm:intro-hypergraph-fpt-to-qp} then gives a quasipolynomial-time algorithm unconditionally. Together with the assumption used to exclude polynomial time, $\NP\nsubseteq\QP$ would make the problem NP-intermediate. We discuss this open problem in \Cref{sec:stoc-open-problems}.

\paragraph{3) Exact counting and narrowing the polymatroid width hierarchy}
It is open in general whether fixed-parameter tractability of deciding
$A\to B$ implies fixed-parameter tractability of counting all
homomorphisms from $A$ to $B$. We settle this question for
bounded $\mu$ and bounded $\alpha$ hypergraph classes and obtain the stronger conclusion that both are equivalent and hold under exactly the same conditions under which it is possible to decide homomorphism existence in polynomial time. Write $\mathsf{\#Hom}(\cH)$ for the problem of counting all homomorphisms from patterns in $\cH$, and write $p\text{-}\mathsf{\#Hom}(\cH)$ for its parameterisation by $\lVert H_A\rVert$. For every recursively enumerable class $\cH$ of bounded $\mu$ or $\alpha$, assuming $\FPT \ne \wone$,
\[
 \pHom(\cH)\in\FPT
 \quad\Longleftrightarrow\quad
 p\text{-}\mathsf{\#Hom}(\cH)\in\FPT
 \quad\Longleftrightarrow\quad
 \mathsf{\#Hom}(\cH)\in\PTIME.
\]

Adaptive and submodular width belong to a larger family of decomposition widths obtained by varying the admissible bag-cost functions. These include widths defined by normal, entropic, and almost-entropic polymatroids, as well as sharp-submodular width, which arises in algorithms for semiring aggregation~\cite{KhamisEtAl2019FAQAI,KhamisEtAl2024LpBounds,Suciu2023InformationInequalities}. All these widths lie between adaptive width and fractional hypertree width. On classes of bounded incidence degeneracy or bounded primal independence, our theorem therefore places the entire hierarchy within an $O(\log\max\{2,\adw(H)\})$ factor of adaptive width. In particular, boundedness of any one of these widths implies boundedness of all the others. We give the definitions and consequences in \Cref{sec:discussion}.

\medskip
\paragraph{Regarding the Exponential Time Hypothesis}
A central part of our results is the characterisation that $\pHom(\cH)\in \FPT$ if and only if $\cH$ has bounded submodular width due to Marx~\cite{Marx13}. Specifically, in the setting of \Cref{thm:intro-width} this implies that unbounded $\fhw$ of $\cH$ implies   $\pHom(\cH)\not\in \FPT$. However, the original result is conditional on the Exponential Time Hypothesis (ETH) which has very recently potentially been refuted~\cite{openaieth,anthropiceth}. As an additional contribution we also prove that Marx's characterisation holds unchanged also under $\FPT \ne \wone$. In particular, we show the following by reduction from the standard parameterised clique problem.
\begin{restatable}[{Marx's theorem~\cite{Marx13} under $\FPT \ne \wone$}]{theorem}{wonehardnesstheorem}
\label{thm:wone}
Let $\cH$ be a recursively enumerable class of hypergraphs.
If $\cH$ has unbounded submodular width, then
$\pHom(\cH)$ is $\wone$-hard.
Consequently, assuming $\FPT\ne\wone$,
\[
 \pHom(\cH)\in\FPT
 \quad\Longleftrightarrow\quad
 \sup_{H\in\cH}\subw(H)<\infty.
\]
\end{restatable}
We prove this by first showing that unbounded submodular width implies that it is always possible to ``efficiently'' embed a clique into a hypergraph in $\cH$ in the sense of the clique embedding power of Fan, Koutris, and Zhao~\cite{stocCliqueEmbedding2023}. Roughly, this means that for some $H \in \cH$, there is a map from $K_m$ into $2^{V(H)}$ such that the images are connected, and image sets touch pairwise. Moreover, if $d$ is the maximum number of clique vertices whose images meet a single hyperedge, the ratio $m/d$ can be made arbitrarily large. Thus we can embed arbitrarily large cliques into members of $\cH$ ``efficiently.''

We combine these embeddings with the local assignment encoding of binary CSPs into homomorphism instances already used in Marx's original proof~\cite{Marx13}. However, to test whether a graph $G$ on $n$ vertices contains a $k$-clique, directly encoding the vertices of $G$ in the established way would give relation tables of size $n^{O(d)}$. Instead, we encode each vertex by $s=\lceil d/2\rceil$ digits over an alphabet of size $O(n^{1/s})$. This ultimately gives an \FPT{} reduction, the details of which are given in \Cref{sec:hardness}.

\paragraph{Organisation of the paper.}
Beyond this introduction, the goal of our initial 12 pages is to prove \Cref{thm:intro-width}. The most important consequences are already informally derived in the introduction, but full details and further consequences that may be of more specific interest (e.g., the distinction between considering classes of hypergraphs and classes of structures as the problem parameter) are discussed afterwards. We choose to discuss our conclusions at the very end of the document as the accompanying discussion of open problems builds on the later technical sections.

Concretely, \Cref{sec:preliminaries} introduces technical preliminaries and the width measures used throughout the paper.
\Cref{sec:separators} proves \Cref{thm:intro-width}.
\Cref{sec:consequences} derives the complexity-theoretic consequences for hypergraph classes and extends them to classes of structures.
\Cref{sec:discussion} develops the consequences for exact homomorphism counting and the polymatroid width hierarchy.
\Cref{sec:ethsave} proves \Cref{thm:wone}.
\Cref{sec:related-work} discusses the relationship of our results to closely related works in the literature.
\Cref{sec:conclusion} concludes with open problems.
\section{Preliminaries}
\label{sec:preliminaries}
All logarithms are to base two unless otherwise specified.

\subsection{Relational Structures and Hypergraphs}

A hypergraph is a pair $H=(V(H),E(H))$, with $E(H) \subseteq 2^{V(H)}$. We refer to $V(H)$ as the vertices of $H$ and $E(H)$ as the edges of $H$. All hypergraphs are finite, have no empty edges, and have every vertex contained in an edge. We assume input hypergraphs are nonempty and connected, but allow induced subhypergraphs to be empty or disconnected. The \emph{rank} $\rank(H)$ is
the maximum hyperedge size.  The \emph{(vertex) degree} $\Delta(H)$ is
the largest number of hyperedges containing one vertex.

For $W\subseteq V(H)$, the \emph{induced subhypergraph} $H[W]$ has vertex
set $W$ and hyperedge set
\[
 E(H[W])=\{e\cap W \mid e\in E(H),\ e\cap W\ne\emptyset\}.
\]
Note that this differs from other common subhypergraph notions where partial edges are not retained. In the context of relational structures, this corresponds to projecting the structure to some subset of the domain elements.

A finite relational signature $\tau$ assigns an arity $r_R$ to each relation
symbol $R$. A finite $\tau$-structure $A$ consists of a finite domain
$V(A)$ and a relation $R^A\subseteq V(A)^{r_R}$ for every $R\in\tau$.
Relations are represented by explicit lists of tuples, and $\lVert A\rVert$
denotes the encoding length. A homomorphism $h:A\to B$ is a map
$V(A)\to V(B)$ such that
\[
 (a_1,\ldots,a_{r_R})\in R^A
 \quad\Longrightarrow\quad
 (h(a_1),\ldots,h(a_{r_R}))\in R^B
\]
for every $R\in\tau$.

A relational structure naturally induces a hypergraph that captures the connections between its elements. The hypergraph $H_A$ has one vertex for each domain element occurring in a tuple. For every tuple $\mathbf a\in R^A$, it has the hyperedge consisting of the entries of $\mathbf a$. Different tuples that induce the same set contribute one hyperedge.

The \emph{primal graph} $P(H)$ has vertex set $V(H)$ and joins two distinct
vertices when they occur together in a hyperedge. Write $\alpha(H)$ for the
independence number of $P(H)$, i.e., the maximum number of vertices such that no two vertices are in the same hyperedge. The \emph{incidence graph} $I(H)$ is the bipartite graph with
parts $V(H)$ and $E(H)$ and incidence as adjacency. Define
\[
 \mu(H)=\max_{\substack{F\subseteq I(H)\\V(F)\ne\emptyset}}
          \min_{x\in V(F)}\deg_F(x).
\]
Thus $\mu(H)$ is the least integer $d$ such that every nonempty subgraph
of $I(H)$ has a vertex of degree at most $d$, called the
\emph{incidence degeneracy} of $H$~\cite{MatulaBeck1983}.

\subsection{Tree Decompositions and Width Functions}
A tree decomposition of a hypergraph $H$ is a pair $(T,B)$, where $T$ is a tree and $B : V(T)\to 2^{V(H)}$ such that:
\begin{enumerate}
\item for each $e \in E(H)$ there is a node $u \in V(T)$ such that $e \subseteq B(u)$.
\item for every vertex $v \in V(H)$, the set of nodes $\{u \in V(T) \mid v \in B(u)\}$ is nonempty and forms a connected subgraph of $T$.
\end{enumerate}
Write
$\mathcal T(H)$ for the set of tree decompositions of $H$ and
$\operatorname{width}_b(T,B)=\max_{t\in V(T)}b(B(t))$ for the width of a
decomposition under a function $b:2^{V(H)}\to\R_{\ge0}$. For a nonempty family
$\mathcal F$ of such functions, its \emph{$\mathcal F$-width} is
\begin{equation}
 \operatorname{width}_{\mathcal F}(H)
 :=\sup_{b\in\mathcal F}\min_{(T,B)\in\mathcal T(H)}
                         \max_{t\in V(T)}b(B(t)).
 \label{eq:F-width}
\end{equation}

For $X\subseteq V(H)$, the \emph{fractional edge-cover number} $\rho_H^*(X)$ is the minimum of
$\sum_e y_e$ over $y_e\ge0$ satisfying $\sum_{e\ni v}y_e\ge1$ for every
$v\in X$. Requiring $y_e\in\{0,1\}$ gives the integral edge-cover number
$\rho_H(X)$. Fractional and generalised hypertree width are
\[
 \fhw(H)=\operatorname{width}_{\{\rho_H^*\}}(H),
 \qquad
 \ghw(H)=\operatorname{width}_{\{\rho_H\}}(H).
\]
Treewidth is
$\operatorname{width}_{\{X\mapsto|X|\}}(H)-1$.
A fractional hypertree decomposition (FHD) supplies each bag with a
fractional edge cover. We use optimum bag covers and write
$\operatorname{width}(\mathcal D)=\operatorname{width}_{\rho_H^*}(T,B)$
for $\mathcal D=(T,B)$.

A \emph{polymatroid} is a function $b:2^{V(H)}\to\R_{\ge0}$ with
$b(\emptyset)=0$ that is monotone and submodular. That is,
\[
 b(X)\le b(Y) \quad\text{whenever }X\subseteq Y,
 \qquad\text{and}\qquad
 b(X\cup Y)+b(X\cap Y) \le b(X)+b(Y)
\]
for all $X,Y\subseteq V(H)$.
It is \emph{edge-dominated} if $b(e)\le1$ for every $e\in E(H)$.
Let $\mathcal S_H$ be the family of edge-dominated polymatroids and let
$\mathcal M_H\subseteq\mathcal S_H$ consist of the modular functions
$b(X)=\sum_{v\in X}w_v$, where $w_v\ge0$ and
$\sum_{v\in e}w_v\le1$ for every hyperedge $e$. We call such a vector $w$
an \emph{edge-dominated modular weighting}. Adaptive width and submodular
width are defined by~\cite{adw,Marx13}
\[
 \adw(H)=\operatorname{width}_{\mathcal M_H}(H),
 \qquad
 \subw(H)=\operatorname{width}_{\mathcal S_H}(H).
\]

It is well known that
$\adw(H)\le\subw(H)\le\fhw(H)\le\ghw(H)$~\cite{adw,Marx13,fhw}.
For $X\subseteq W\subseteq V(H)$, fractional cover cost satisfies
\begin{equation}
 \rho_{H[W]}^*(X)=\rho_H^*(X).
 \label{eq:induced-cover-cost}
\end{equation}
Indeed, a fractional cover of $H$ induces one of $H[W]$ by aggregating the
weights of all hyperedges with the same nonempty intersection with $W$ and
discarding hyperedges disjoint from $W$. Conversely, a fractional cover of
$H[W]$ lifts to $H$ by choosing, for every $f\in E(H[W])$, one hyperedge
$e_f\in E(H)$ with $e_f\cap W=f$ and assigning the weight of $f$ to $e_f$.
The first operation does not increase total weight, and the second preserves
it, which proves \eqref{eq:induced-cover-cost}.
Moreover, it is straightforward that adaptive width is nonincreasing under taking induced subhypergraphs, since modular
weightings extend by zero and tree decompositions restrict to vertex subsets.
\section{Fractional Separators and Adaptive Width}
\label{sec:separators}

Abo Khamis, Ngo, and Suciu~\cite{KhamisNgoSuciu2025PANDA} observed that fractional edge-cover duality lets us naturally view adaptive width and fractional hypertree
width as opposite orders of optimisation.
\[
 \adw(H)=\max_{b\in\mathcal M_H}\min_{(T,B)\in\mathcal T(H)}
                  \operatorname{width}_b(T,B),
 \qquad
 \fhw(H)=\min_{(T,B)\in\mathcal T(H)}\max_{b\in\mathcal M_H}
                  \operatorname{width}_b(T,B).
\]
Here $\mathcal M_H$ is the family of edge-dominated modular functions.
Weak minimax recovers only the familiar inequality
$\adw(H)\le\fhw(H)$.
Strong minimax does not apply directly. The space of tree decompositions has
no useful convex structure, and
$\operatorname{width}_b(T,B)=\max_{t\in V(T)}b(B_t)$ is the maximum of
linear functions of $b$, hence convex rather than concave in the
maximising variable.

The difference between the two optimisation orders is that adaptive
width allows us to choose a different decomposition for each $b$, whereas
fractional hypertree width asks for one decomposition whose bags are cheap
for every $b$. Our proof bridges this difference through balanced
separators. They are simpler objects than decompositions, but retain the
two properties we need. A decomposition supplies a balanced separator,
and suitable balanced separators can be assembled recursively into a
decomposition. We will show that, for each fixed hyperedge weighting and balance threshold, the fractional balanced separators form a compact convex set. Together with fractional-cover duality, this will allow us to apply minimax and bound their cost by adaptive width.

In \Cref{sec:separator-definitions}, we recall fractional balanced separators and their fractional edge-cover cost.
In \Cref{sec:adaptive}, we use minimax to show that adaptive width bounds the cost of such fractional balanced separators in every induced subhypergraph.
In \Cref{sec:rounding}, we round these separators and assemble them
recursively into a tree decomposition, proving \Cref{thm:intro-width}.

\subsection{Fractional Balanced Separators}
\label{sec:separator-definitions}

We first recall the hypergraph fractional balanced-separator relaxation
from~\cite[Definition~5.2]{KorchemnaEtAl2024}. It is a hyperedge-weighted
analogue of the spreading-metric formulation of $\rho$-separators due to Even,
Naor, Rao, and Schieber~\cite{EvenNaorRaoSchieber1999}, building on the metric
and region-growing framework of Leighton and Rao~\cite{LeightonRao1999}.

In the following, we consider vertex weightings $x : V(J) \to [0,1]$ that we interpret as \emph{fractional separators} on a hypergraph $J$. This relaxes the classical notion of a separator, in which vertices either belong to the separator or do not. For $S \subseteq V(J)$, we write $\mathbf 1_S$ for the indicator function that maps vertices in $S$ to $1$ and all other vertices to $0$. For a vertex weighting $x$, define
\[
 \rho_J^*(x)
 :=\min\left\{
   \sum_{e\in E(J)}y(e):
   y:E(J)\to\mathbb R_{\ge0},\quad
   \sum_{e\ni v}y(e)\ge x(v)\ \forall v\in V(J)
 \right\}.
\]
Thus a fractional edge cover need only satisfy $\sum_{e\ni v}y(e)\ge x(v)$, rather than $\sum_{e\ni v}y(e)\ge 1$ uniformly. When $x$ is a fractional separator, this will eventually correspond to the ``cost'' of the separator.

For vertices $u,v$, let $\mathsf{paths}(P(J);u,v)$ be the set of simple
paths in $P(J)$ with endpoints $u$ and $v$.
We are interested in the minimum-weight path between any pair of vertices. To that end, define
\[
 \operatorname{dist}_{J,x}(u,v)
 :=\min_{Q\in\mathsf{paths}(P(J);u,v)}\sum_{z\in V(Q)}x(z),
 \qquad
 \operatorname{dist}_{J,x}^*(u,v)
 :=\min\{1,\operatorname{dist}_{J,x}(u,v)\}.
\]
Note that the sum includes both endpoints and we consider the minimum of the empty set to be $\infty$. When $x$ is seen as a fractional separator, the distance between $u$ and $v$ is a measure of how separated they are, with values above 1 meaning they are fully separated.
We therefore use the auxiliary quantity $\operatorname{dist}_{J,x}^*$, which is the truncation of $\operatorname{dist}_{J,x}$ to $[0,1]$. Note that if
$x=\mathbf 1_S$ for some $S\subseteq V(J)$, then
$\operatorname{dist}_{J,x}^*(u,v)=0$ exactly when $u$ and $v$ are connected in $P(J)-S$, and it equals $1$ otherwise.
We extend this distance to hyperedges naturally. For hyperedges $e,f$, define
\(
 \operatorname{dist}_{J,x}^*(e,f)
 :=\min_{u\in e,\,v\in f}\operatorname{dist}_{J,x}^*(u,v).
\)

We also consider a weighting $\gamma:E(J)\to[0,1]$ that records the contribution of each edge to the balance condition.
In particular, let $\gamma:E(J)\to[0,1]$ be nonzero and write
$\Gamma=\sum_e\gamma(e)$.  For $\theta\in(0,1)$, a function
$x:V(J)\to[0,1]$ is a \emph{fractional $(\gamma,\theta)$-balanced separator}
if
\begin{equation}
 \sum_{f\in E(J)}
 \gamma(f)\operatorname{dist}_{J,x}^*(e,f)
 \ge (1-\theta)\Gamma
 \qquad\forall e\in E(J).
 \label{eq:fractional-balance}
\end{equation}
Intuitively, the left-hand side can be read as a measure of how strongly edge $e$ is separated from all other edges $f$ (weighted by $\gamma$) by the fractional separator $x$.

When the separator is not fractional, i.e., it is a set $S\subseteq V(J)$, we say that $S$ is \emph{$(\gamma,\theta)$-balanced} if
every component $C$ of $J-S$ satisfies
\begin{equation}
 \sum_{\substack{e\in E(J)\\e\cap C\ne\emptyset}}\gamma(e)
 \le \theta\Gamma.
 \label{eq:integral-balance}
\end{equation}

\begin{lemma}[{\cite[Fact~5.4]{KorchemnaEtAl2024}}]
\label{lem:indicator-equivalence}
Let $J$ be a hypergraph, $S\subseteq V(J)$,
$\gamma:E(J)\to[0,1]$ nonzero, and $\theta\in(0,1)$.
Then $\mathbf 1_S$ is a fractional
$(\gamma,\theta)$-balanced separator if and only if $S$ is $(\gamma,\theta)$-balanced.
\end{lemma}

So far, we have followed the definitions of~\cite{KorchemnaEtAl2024}. The linear-programming dual form of $\rho^*$ will also be important for our result. It is
\begin{equation}
 \rho_J^*(x)
 =\max\left\{
  \sum_{v\in V(J)}w_vx(v):
  w_v\ge0,\quad
  \sum_{v\in e}w_v\le1\ \forall e\in E(J)
 \right\}.
 \label{eq:cover-dual}
\end{equation}
Moreover, we will be interested in the least $\rho^*$ necessary for a $(\gamma,\theta)$-balanced fractional separator, for fixed $\gamma,\theta$. Define the optimum fractional cost of one separator instance by
\[
 \beta_{J,\theta}(\gamma)
 :=
 \min\left\{
  \rho_J^*(x) \mid
  x\text{ is a fractional }(\gamma,\theta)\text{-balanced separator}
 \right\}.
\]
Finally, we use fractional balanced separators as an intermediate object between
adaptive width and fractional hypertree width. Specifically, we will bound their cost in every induced subhypergraph and for every weighting of its hyperedges.  To state this stronger bound, we define
\[
 \fsep_\theta(H)
 :=
 \sup_{\substack{W\subseteq V(H)\\
                  0\ne\gamma:E(H[W])\to[0,1]}}
 \beta_{H[W],\theta}(\gamma).
\]
Equivalently, $\fsep_\theta(H)$ is the least $b$ such that every induced
subhypergraph $J$ of $H$ and every nonzero hyperedge weighting $\gamma$
admit a fractional $(\gamma,\theta)$-balanced separator $x$ with
$\rho_J^*(x)\le b$.

\subsection{Adaptive Width and Fractional Separators}
\label{sec:adaptive}

The convexity of fractional balanced separators allows minimax to compare
their optimum cost with adaptive width. We prove
$\fsep_{1/2}(H)\le\adw(H)$.

Fix a hypergraph $J$ and a nonzero hyperedge weighting $\gamma$.
Every tree decomposition of $J$ contains a $(\gamma,1/2)$-balanced bag,
as we will see below. To prove
$\beta_{J,1/2}(\gamma)\le\adw(J)$, it therefore suffices to find one
$b\in\mathcal M_J$ under which every such bag has cost at least
$\beta_{J,1/2}(\gamma)$.

Fractional-cover duality already gives us a modular weighting witnessing
the cost of each separator, but this weighting may depend on the separator.
We need one weighting that works for all of them. The fractional relaxation
allows us to obtain it by a classic minimax argument, because the feasible
separators form a compact convex set.

\begin{lemma}
\label{lem:fractional-separator-convexity}
Let $J$ be a hypergraph, let $\gamma:E(J)\to[0,1]$ be nonzero, and let
$\theta\in(0,1)$. The set of fractional
$(\gamma,\theta)$-balanced separators is nonempty, compact, and convex.
\end{lemma}

\begin{proof}
For $e,f\in E(J)$, let $\mathcal P(e,f)$ be the set of simple paths in
$P(J)$ joining a vertex of $e$ to a vertex of $f$, including paths of
length zero. By definition,
\[
 \operatorname{dist}_{J,x}^*(e,f)
 =\min\left(
    \{1\}\cup
    \left\{\sum_{v\in V(Q)}x(v)\mid Q\in\mathcal P(e,f)\right\}
   \right).
\]
Thus $x\mapsto\operatorname{dist}_{J,x}^*(e,f)$ is continuous and concave,
since it is the minimum of finitely many affine functions. Since $\gamma\ge0$,
the left-hand side of each balance constraint~\eqref{eq:fractional-balance}
is also continuous and concave. Each constraint therefore defines a closed
convex subset of $[0,1]^{V(J)}$. Their intersection is convex and is a closed
subset of the compact cube $[0,1]^{V(J)}$, hence compact. It is nonempty
because it contains the all-one vector.
\end{proof}

We can now combine this convexity with fractional-cover duality to obtain
the desired common weighting. The argument in fact works for every balance
threshold $\theta\in(0,1)$, although we will only use it with $\theta=1/2$.

\begin{lemma}
\label{lem:cover-minimax}
Let $J$ be a hypergraph, let $\gamma:E(J)\to[0,1]$ be nonzero, and let
$\theta\in(0,1)$.  There
is $b \in \mathcal{M}_J$ such that
\[
 \sum_{v\in V(J)}b(\{v\})x(v)
 \ge \beta_{J,\theta}(\gamma)
\]
for every fractional $(\gamma,\theta)$-balanced separator $x$.
\end{lemma}

\begin{proof}
Let $\mathcal X$ be the set of fractional
$(\gamma,\theta)$-balanced separators. Identify each $b\in\mathcal M_J$
with its singleton values $(b(\{v\}))_{v\in V(J)}$. By \Cref{eq:cover-dual},
\[
 \beta_{J,\theta}(\gamma)
 =\min_{x\in\mathcal X}\max_{b\in\mathcal M_J}\sum_vb(\{v\})x(v).
\]
By \Cref{lem:fractional-separator-convexity}, $\mathcal X$ is nonempty,
compact, and convex. In these coordinates, $\mathcal M_J$ is a closed convex
polytope containing the zero vector. Every vertex $v$ lies in some
hyperedge $e$, so $0\le b(\{v\})\le b(e)\le1$. Thus
$\mathcal M_J\subseteq[0,1]^{V(J)}$, and $\mathcal M_J$ is compact. Furthermore, the map
$(x,b)\mapsto\sum_vb(\{v\})x(v)$ is continuous and bilinear. A standard
application of the minimax theorem over convex sets~\cite{Sion58} thus gives
\begin{equation}
 \beta_{J,\theta}(\gamma)
 =\max_{b\in\mathcal M_J}\min_{x\in\mathcal X}\sum_vb(\{v\})x(v).
 \label{eq:separator-minimax}
\end{equation}
Compactness and continuity ensure that the maximum is attained.
Choose a maximiser $b$. Then every $x\in\mathcal X$ satisfies
\[
 \sum_v b(\{v\})x(v)
 \ge \min_{y\in\mathcal X}\sum_v b(\{v\})y(v)
 =\beta_{J,\theta}(\gamma),
\]
as required.
\end{proof}

What remains is to put everything together and show that every tree
decomposition contains a $(\gamma,1/2)$-balanced bag. This follows by a
standard weighted-centroid argument on the decomposition tree, with the
hyperedge weights assigned to bags containing the respective edges.
The preceding lemma therefore forces every decomposition to contain a bag
of $b$-cost at least $\beta_{J,1/2}(\gamma)$, with the same $b$ working for
all decompositions. Applying this argument in every induced subhypergraph
yields the desired bound on $\fsep_{1/2}(H)$.

\begin{lemma}
\label{thm:fsep-adw}
For every hypergraph $H$,
\(
 \fsep_{1/2}(H)\le\adw(H).
\)
\end{lemma}

\begin{proof}
Fix $W\subseteq V(H)$, let $J=H[W]$, and fix a nonzero weighting
$\gamma:E(J)\to[0,1]$. Choose $b\in\mathcal M_J$ from
\Cref{lem:cover-minimax}, so that
$\sum_{v\in V(J)}b(\{v\})x(v)\ge\beta_{J,1/2}(\gamma)$ for every fractional
$(\gamma,1/2)$-balanced separator $x$.
We will show that every tree decomposition of $J$ has a bag of $b$-cost
at least $\beta_{J,1/2}(\gamma)$.

Let $(T,B)$ be any tree decomposition of $J$. For each $e\in E(J)$,
choose one node $a(e)\in V(T)$ with $e\subseteq B_{a(e)}$. Give each
node $s$ weight $q(s)=\sum_{e\in E(J),\,a(e)=s}\gamma(e)$, adding the
weights of all hyperedges assigned to $s$. Writing
$\Gamma=\sum_{e\in E(J)}\gamma(e)$, we have
$\sum_{s\in V(T)}q(s)=\Gamma$.

We first show that there is a node $t$ such that every component $K$ of $T-t$ satisfies
$$\sum_{s\in V(K)}q(s)\le\Gamma/2.$$ To construct $t$, start at any
node $s$. Whenever a component of $T-s$ has total $q$-weight greater
than $\Gamma/2$, move to the neighbour of $s$ in that component.
After a move from $s$ to $s'$, the component of $T-s'$ containing $s$
has weight less than $\Gamma/2$, so the next move cannot return to $s$.
A walk in a tree that never immediately reverses an edge cannot revisit
a node. Since $T$ is finite, the process terminates at a node $t$ with
the required property.

We next show that $B_t$ is $(\gamma,1/2)$-balanced. Let $C$ be the vertex
set of a component of $P(J)-B_t$. We claim that all nodes whose bags meet
$C$ lie in a single component $K$ of $T-t$. For each $v\in C$, the set
$T_v=\{s\in V(T)\mid v\in B_s\}$ is connected and avoids $t$, so it
lies in one component of $T-t$. If $u,v\in C$ are adjacent in $P(J)$,
some hyperedge contains both, and a bag containing that hyperedge
belongs to $T_u\cap T_v$. Since $C$ is connected, this proves the claim.
For every hyperedge
$e$ meeting $C$, choose $v\in e\cap C$. Then
$a(e)\in T_v\subseteq V(K)$, since $e\subseteq B_{a(e)}$. Therefore
$$\sum_{e\in E(J),\,e\cap C\ne\emptyset}\gamma(e)
\le\sum_{s\in V(K)}q(s)\le\Gamma/2$$
as required by \eqref{eq:integral-balance}.

By \Cref{lem:indicator-equivalence}, $\mathbf 1_{B_t}$ is a fractional
$(\gamma,1/2)$-balanced separator. The choice of $b$ gives
$b(B_t)=\sum_{v\in V(J)}b(\{v\})\mathbf 1_{B_t}(v)\ge\beta_{J,1/2}(\gamma)$.
Since $(T,B)$ was arbitrary and $b\in\mathcal M_J$, we obtain
\[
 \beta_{J,1/2}(\gamma)
 \le\min_{(T,B)\in\mathcal T(J)}\max_{s\in V(T)}b(B_s)
 \le\adw(J)\le\adw(H).
\]
Taking the supremum over $W$ and $\gamma$
gives $\fsep_{1/2}(H)\le\adw(H)$.
\end{proof}

We view the simplicity of this argument as a consequence of the accumulated work on these width measures. Especially the minimax formulation in the work on PANDA~\cite{KhamisNgoSuciu2025PANDA} and the fractional balanced separators of~\cite{KorchemnaEtAl2024} seem to have achieved the right level of abstraction where everything lines up and the arguments of this section are fairly natural. The remaining step is to return from fractional separators to a tree decomposition, using the rounding theory developed for these fractional balanced separators.

\subsection{Rounding Fractional Separators and Decompositions}
\label{sec:rounding}

We will next use the rounding theorem of~\cite{KorchemnaEtAl2024} to convert fractional balanced separators into integral separators with controlled fractional edge-cover number. Combined with $\fsep_{1/2}(H)\le\adw(H)$, this then supplies the actual separators needed to produce a decomposition with $\fhw$ bounded in terms of $\adw$.

\begin{proposition}[{\cite[Theorem~5.6]{KorchemnaEtAl2024}}]
\label{thm:rounding}
Let $J$ be a hypergraph, let $0\ne\gamma:E(J)\to[0,1]$, and let $x$ be a
fractional $(\gamma,1/2)$-balanced separator. There is an integral
$(\gamma,5/6)$-balanced separator $S\subseteq V(J)$ satisfying
\[
 \begin{aligned}
 \rho_J^*(S)&\le C_0\mu(J)\rho_J^*(x)\log\max\{2,\rho_J^*(x)\},\\
 \rho_J^*(S)&\le C_0\max\{1,\log\alpha(J)\}\rho_J^*(x)\log\max\{2,\rho_J^*(x)\}
 \end{aligned}
\]
for a universal constant $C_0$.
\end{proposition}

\begin{proof}
If $\emptyset$ is $(\gamma,5/6)$-balanced, take $S=\emptyset$.
Otherwise set $r=\rho_J^*(x)$. If $r=0$, then $x=0$, so
$\mathbf 1_\emptyset=x$ is fractionally $(\gamma,1/2)$-balanced.
By \Cref{lem:indicator-equivalence}, $\emptyset$ is integrally
$(\gamma,1/2)$-balanced and hence $(\gamma,5/6)$-balanced, a contradiction.
Thus $r>0$.
Write $\supp x=\{v\in V(J)\mid x(v)>0\}$.
The cited theorem with $\varphi=1/2$ produces a $(\gamma,5/6)$-balanced
separator and gives
\[
 \rho_J^*(S)
 \le
 \bigl(\min\{8+4\ln\alpha(J[\supp x]),6\mu(J)\}+1\bigr)
 \bigl(88+16\log(2r)\bigr)r.
\]
Put $A=\alpha(J[\supp x])$. Then $A\le\alpha(J)$. Since $\mu(J)\ge1$,
the first factor is at most $7\mu(J)$ and at most
$9+4\ln\alpha(J)$. Moreover,
\[
 88+16\log(2r)
 \le 88+16\max\{0,\log(2r)\}
 =O(\log\max\{2,r\}).
\]
Absorbing the constants into $C_0$ proves the claim.
\end{proof}

The dependence on incidence degeneracy and primal independence enters only through the rounding step. Together with the preceding section we thus get integral $(\gamma,5/6)$-balanced separators, whose fractional edge-covers can be bounded in terms of adaptive width. The last step is then the relatively standard process of turning balanced separators into tree decompositions. For the $(\gamma,5/6)$-balanced separators specifically we directly use the following lemma of Marx.

\begin{lemma}[{\cite[Section~5.2]{adw}}]
\label{lem:recursive-decomposition}
Let $H$ be a hypergraph, let $\eta\in(0,1)$, let $r>0$, and let
$f:2^{V(H)}\to\R_{\ge0}$ be monotone and subadditive.
Suppose that for every $X\subseteq V(H)$ there is a set $S\subseteq V(H)$
with $f(S)\le r$ such that every component $C$ of $H-S$ satisfies
\[
 f(C\cap X)\le\eta f(X).
\]
Then $H$ has a tree decomposition $(T,B)$ with
\[
 f(B_t)\le\frac{2r}{1-\eta}+r
 \qquad\text{for every }t\in V(T).
\]
\end{lemma}

\mainwidththeorem*

\begin{proof}[Proof of \Cref{thm:intro-width}]
Set $a=\max\{2,\adw(H)\}$ and put $f(X)=\rho_H^*(X)$.
Then $f$ is monotone and subadditive. Fix a nonempty $X\subseteq V(H)$
and choose an optimum fractional edge cover $\gamma:E(H)\to[0,1]$ of $X$.
Since $X\ne\emptyset$, the weighting $\gamma$ is nonzero. Let $x$ be a
fractional $(\gamma,1/2)$-balanced separator satisfying
\(
 \rho_H^*(x)=\beta_{H,1/2}(\gamma).
\)
Such an $x$ exists by \Cref{lem:fractional-separator-convexity} and
continuity of $\rho_H^*$. By the definition of $\fsep_{1/2}$ and
\Cref{thm:fsep-adw},
\[
 \rho_H^*(x)
 =\beta_{H,1/2}(\gamma)
 \le\fsep_{1/2}(H)
 \le\adw(H)
 \le a.
\]
As $a\ge2$, \Cref{thm:rounding} tells us that
\[
 \rho_H^*(S)\le C_0\mu(H)a\log a,
 \qquad
 \rho_H^*(S)\le C_0\max\{1,\log\alpha(H)\}a\log a
\]
for some $(\gamma,5/6)$-balanced separator $S$.
For every component $C$ of $H-S$, restricting $\gamma$ to edges meeting
$C$ still covers $X\cap C$. The $(\gamma,5/6)$-balance of $S$ therefore gives
\[
 f(X\cap C)
 \le\sum_{e\cap C\ne\emptyset}\gamma(e)
 \le\frac56\sum_{e\in E(H)}\gamma(e)
 =\frac56 f(X).
\]
For $X=\emptyset$, take $S=\emptyset$.
Thus \Cref{lem:recursive-decomposition} applies with $\eta=5/6$ and
$r=C_0\mu(H)a\log a$ and $r=C_0\max\{1,\log\alpha(H)\}a\log a$, respectively.
Since $2r/(1-5/6)+r=13r$, we obtain
\[
 \fhw(H)\le C\mu(H)a\log a,
 \qquad
 \fhw(H)\le C\max\{1,\log\alpha(H)\}a\log a.
\]
For $\adw(H)\ge2$ these are the stated bounds,
and for $\adw(H)<2$ they give $\fhw(H)=O(\mu(H))$ and $\fhw(H)=O(\max\{1,\log\alpha(H)\})$.
\end{proof}

Since $\alpha(H)\le |V(H)|$, the independence bound gives the following corollary.

\begin{corollary}
\label{cor:stoc-universal-bridge}
For every $n$-vertex hypergraph $H$ with $\adw(H)\ge2$,
\[
 \fhw(H)\le C\adw(H)\log\adw(H)\log n
\]
for a universal constant $C$.
\end{corollary}

Bressan, Lanzinger, and Roth~\cite{BressanLanzingerRoth2023} already observed a closely related statement to the $\alpha$ part of \Cref{thm:intro-width}. Towards a different problem, they prove the integrality-gap bound
\[
 \alpha(H)\ge\frac12+\frac{\alpha^*(H)}{4\adw(H)},
\]
where $\alpha^*(H)$ is the fractional independence number of $H$.
While they did not observe the connections to fractional hypertree width or the complexity theoretic consequences, by linear programming duality, $\alpha^*(H)=\rho_H^*(H)$, so their bound also shows
\[
 \rho_H^*(V(H))\le (4\alpha(H)-2)\adw(H).
\]
This is incomparable to the bound from \Cref{thm:intro-width} in a number of interesting ways. First, this shows that the entire global fractional cover number is in $O(\alpha(H) \adw(H))$, which then yields the same bound for $\fhw$ as a special case. This means that not only are the individual separators boundable, but the interaction of adaptive width and independence number has a global implication for $H$ that is yet to be explored further. Moreover, no additional $\log \adw$ factor is required in this bound.

However, since the dependence on $\alpha$ is linear, this bound is also not strong enough to obtain only logarithmic dependence on $n$ in \Cref{cor:stoc-universal-bridge}. The logarithmic dependence on $\alpha$ in \Cref{thm:intro-width} is what ultimately leads to what we consider the most surprising result presented here, namely, \Cref{thm:intro-hypergraph-fpt-to-qp}.

\section{Complexity Consequences}
\label{sec:consequences}

We now translate \Cref{thm:intro-width} into complexity statements for
homomorphism problems. For hypergraph classes, fixed-parameter tractability is characterised by bounded submodular width, while the standard polynomial-time algorithms require bounded fractional hypertree width. We show how \Cref{thm:intro-width}  yields an $\FPT$--$\PTIME$
collapse whenever incidence degeneracy or independence number is bounded,  and a quasipolynomial algorithm for
every $\FPT$ hypergraph class. We then additionally consider the situation for classes of relational structures rather than classes of hypergraphs. There the relevant widths are those of their cores (the minimal endomorphic images of the pattern structure) but computing the core itself is intractable.

\subsection{Hypergraph Classes}

First, recall that, given an FHD $\mathcal D$ of width $k$ for $H_A$, deciding $A\to B$
takes time
$\operatorname{poly}(\lVert A\rVert+|\mathcal D|+\lVert B\rVert)\,\lVert B\rVert^{O(k)}$~\cite{fhw}. Computing an optimum FHD is $\NP$-hard even when the width is bounded by a fixed constant~\cite{DBLP:journals/jacm/GottlobLPR21}. Moreover, the size of the FHD can be assumed without loss of generality to be in $O(\lVert A \rVert)$. However, Marx's cubic approximation computes an FHD of width $O(k^3)$ in time $\lVert H_A\rVert^{O(k^3)}$ whenever $\fhw(H_A)\le k$~\cite[Theorem~4.1]{MarxFHW2010}. At the cost of a cubic overhead in the exponent, one still obtains a polynomial-time algorithm for every fixed $k$.

On classes of bounded incidence degeneracy or bounded independence number, \Cref{thm:intro-width}
and the width hierarchy make bounded fractional hypertree width equivalent
to bounded submodular width, which characterises fixed-parameter
tractability under $\FPT\ne \wone$ by \Cref{thm:wone}. \Cref{thm:intro-width} directly yields the $\FPT=\PTIME$ collapse in this setting. The following more expansive theorem also directly implies \Cref{thm:intro-collapse}.

\begin{theorem}
\label{thm:bounded-lambda-collapse}
Let $\cH$ be a recursively enumerable hypergraph class of bounded incidence degeneracy or bounded independence number. Assuming $\FPT\ne \wone$, the following are equivalent.
\begin{enumerate}[label=(\roman*)]
 \item $\cH$ has bounded fractional hypertree width.
 \item $\cH$ has bounded submodular width.
 \item $\pHom(\cH)\in\FPT$.
 \item $\Hom(\cH)\in\PTIME$.
\end{enumerate}
\end{theorem}

\begin{proof}
The width hierarchy $\adw(H)\le\subw(H)\le\fhw(H)$ together with
\Cref{thm:intro-width} gives (i)$\Leftrightarrow$(ii).
By \Cref{thm:wone},
(ii)$\Leftrightarrow$(iii) under $\FPT\ne \wone$.
For (i)$\Rightarrow$(iv), Grohe and Marx prove directly that
$\Hom(\cH)\in\PTIME$~\cite[Corollary~4.11]{fhw}.
Every polynomial-time algorithm is fixed-parameter tractable, giving
(iv)$\Rightarrow$(iii).
\end{proof}

Since $\alpha(H)\le |V(H)|$, our independence bound gives $\fhw(H)=O(\log N)$ whenever adaptive width is bounded. Combining this with the approximation algorithm of~\cite{KorchemnaEtAl2024} yields the quasipolynomial bound. The additional $(\log\log N)^2$ factor comes from approximating the decomposition. Results on computing fractional hypertree decompositions under structural restrictions, such as those of~\cite{DBLP:journals/jacm/GottlobLPR21}, suggest a way to sharpen this bound. On any class where a decomposition of width $O(k)$ can be computed in time $N^{O(k)}$ whenever $\fhw(H)\le k$, the same argument gives a running time of $N^{O(\log N)}$.
\begin{proof}[Proof of \Cref{thm:intro-hypergraph-fpt-to-qp}]
By \Cref{thm:wone} and $\pHom(\cH)\in\FPT$, $\cH$ has bounded submodular
width, and hence bounded adaptive width.
For an input $(A,B)$ of size $N$, \Cref{thm:intro-width} gives
$\fhw(H_A)=O(\log N)$.
The approximation algorithm of~\cite[Theorem~1.2]{KorchemnaEtAl2024}
therefore constructs an FHD of width $O(\log N\,(\log\log N)^2)$
in time $N^{O(\log N)}$.
Evaluating $(A,B)$ using this decomposition takes total time
\[
 N^{O(\log N\,(\log\log N)^2)}=N^{(\log N)^{1+o(1)}}.\qedhere
\]
\end{proof}

\subsection{Widths of Cores and Structure Classes}
In the following, for a relational structure $A$ and a hypergraph width function $w$,
we write $w(A)$ for $w(H_A)$.

An endomorphism of $A$ is a homomorphism from $A$ to itself.
A \emph{core} of $A$ is an endomorphic image of minimum domain
cardinality. It is unique up to isomorphism, and we denote it by
$\operatorname{core}(A)$ (see also~\cite{ChandraMerlin77}).
Since $A$ and $\operatorname{core}(A)$ are homomorphically equivalent,
\[
 A\to B \quad\Longleftrightarrow\quad \operatorname{core}(A)\to B.
\]
Their widths can nevertheless differ arbitrarily. For example, the
undirected $n\times n$ grid with $n\ge2$ has a single edge as its core, even though its
fractional hypertree width grows with $n$.

Observe that in the $\Hom(\cH)$ setting this has no effect, as every hypergraph $H$ admits a structure $A$ with $H_A=H$ where $A$ is a core. This motivates an alternative version of the problem where we restrict the patterns to a class $\cA$ of finite relational structures.

\begin{problem}{$\Hom(\cA)$}
Input & Finite relational structures $A$ and $B$ over the same signature,
with $A\in\cA$.\\
Question & Does $A\to B$?
\end{problem}

We write $\pHom(\cA)$ for the parameterisation of $\Hom(\cA)$ by the pattern $A$ and write $\operatorname{core}(\cA)$ for the class of cores of structures
in $\cA$.
Grohe~\cite{DBLP:journals/jacm/Grohe07} showed that,
for any recursively enumerable class $\cA$ of structures of bounded arity, bounded treewidth
of $\operatorname{core}(\cA)$ characterises both polynomial-time solvability and
fixed-parameter tractability, assuming $\FPT\ne\wone$.
We extend this characterisation to bounded incidence degeneracy using
fractional hypertree width of the cores.

The main challenge over the hypergraph setting is the upper bound. Recall that the efficient algorithm for deciding $A\to B$ with exponent $O(k)$ requires an explicit width-$O(k)$ decomposition for $A$. Even in the $\Hom(\cH)$ setting, obtaining such a decomposition in polynomial time requires an approximation algorithm. In the $\Hom(\cA)$ setting, this is not feasible as finding the core itself is already intractable~\cite{ChandraMerlin77}.

To avoid this issue, Dalmau, Kolaitis, and Vardi~\cite{DalmauKolaitisVardi2002}
showed how to use existential pebble games to decide $A\to B$ in polynomial time
when the cores have bounded treewidth, without having to compute the cores.
Chen and Dalmau~\cite{ChenDalmau2005} extend this approach to bounded
$\ghw$ of the cores. To apply it to fractional
hypertree width, we need to bound the number of edges needed to cover
each bag, rather than just their total fractional weight.
Chen, Gottlob, Lanzinger, and Pichler~\cite[Section~4]{ChenGottlobLanzingerPichler2020} use
bounds on the integrality gap of edge covers to show that, under bounded
VC dimension, bounded $\fhw$ implies bounded $\ghw$.
This transfers the upper bound without constructing a decomposition of
the core. Bounded incidence degeneracy ensures bounded VC dimension,
giving the following proposition.

\begin{proposition}
\label{prop:sparse-core-fhw-ptime}
Let $\cA$ be a class of finite relational structures of bounded incidence
degeneracy. If
$\operatorname{core}(\cA)$ has bounded fractional hypertree width, then
$\Hom(\cA)\in\PTIME$.
\end{proposition}

\begin{proof}
By \cite[Lemma~2 and Section~4, Corollary~1]{ChenGottlobLanzingerPichler2020},
bounded VC dimension of $\{H_A:A\in\cA\}$ and bounded fractional
hypertree width of $\operatorname{core}(\cA)$ imply $\Hom(\cA)\in\PTIME$.
Let an integer $d$ bound the incidence degeneracy. The incidence graphs
contain no $K_{d+1,d+1}$, so
$\operatorname{VCdim}(H_A)\le d+\lceil\log_2(d+1)\rceil$ for every
$A\in\cA$ by \cite[Table~1 and Appendix~B]{BringmannKozmaMoranNarayanaswamy2016}.
\end{proof}

The proposition requires bounded incidence degeneracy. Whether it also holds under bounded primal independence number remains open.
This remains an interesting gap between the following \Cref{thm:semantic-incidence-collapse} and \Cref{thm:bounded-lambda-collapse}.

\begin{theorem}
\label{thm:semantic-incidence-collapse}
Let $\cA$ be a recursively enumerable class of finite relational structures
with bounded incidence degeneracy. Assuming $\FPT\ne \wone$, the following are equivalent.
\begin{enumerate}[label=(\roman*)]
 \item $\operatorname{core}(\cA)$ has bounded fractional hypertree width.
 \item $\operatorname{core}(\cA)$ has bounded submodular width.
 \item $\pHom(\cA)\in\FPT$.
 \item $\Hom(\cA)\in\PTIME$.
\end{enumerate}
\end{theorem}

\begin{proof}
For $A\in\cA$, set $C=\operatorname{core}(A)$ and let $d$ bound the incidence
degeneracy of $\cA$. We may choose $C$ as a substructure of $A$.
Thus $I(H_{C})$ is a subgraph of $I(H_A)$, and
$\mu(H_{C})\le d$.
For $\subw(C)\ge2$, the width hierarchy and \Cref{thm:intro-width} give
\[
 \subw(C)\le\fhw(C)
 \le O_d\!\left(
 \subw(C)\log\subw(C)
 \right).
\]
For $\subw(C)<2$, the same theorem gives $\fhw(C)=O_d(1)$.
This proves (i)$\Leftrightarrow$(ii). By~\cite[Theorem~1 and Lemma~2]{ChenGottlobLanzingerPichler2020}\footnote{This characterisation was originally stated under ETH, but the same argument immediately applies under $\FPT\ne \wone$ in view of \Cref{thm:wone}.},
$\pHom(\cA)\in\FPT$ if and only if $\operatorname{core}(\cA)$ has bounded
submodular width under $\FPT\ne \wone$, giving (ii)$\Leftrightarrow$(iii).
\Cref{prop:sparse-core-fhw-ptime} gives (i)$\Rightarrow$(iv).
Every polynomial-time algorithm is fixed-parameter tractable, giving
(iv)$\Rightarrow$(iii).
\end{proof}

Bounded incidence degeneracy includes both bounded rank and bounded
vertex degree, and also allows both to be unbounded when the incidence
graphs are sparse. The following corollary gives several standard
settings in which the characterisation applies.

\begin{corollary}
\label{cor:sparse-incidence-classes}
Let $\cA$ be a recursively enumerable class of finite relational structures.
Each of the following conditions implies the equivalences in
\Cref{thm:semantic-incidence-collapse}, assuming $\FPT\ne \wone$.
\begin{enumerate}[label=(\roman*)]
 \item The hypergraphs $H_A$ have bounded rank.
 \item The hypergraphs $H_A$ have bounded maximum vertex degree.
 \item The hypergraphs $H_A$ are Zykov-planar, i.e., the
 incidence graphs $I(H_A)$ are planar (see~\cite{Zykov1974}).
 \item The incidence graphs $I(H_A)$ exclude a fixed graph as a topological minor.
\end{enumerate}
\end{corollary}

\begin{proof}
For every hypergraph $H$, its incidence degeneracy is at most its rank
and at most its maximum vertex degree. Thus (i) and (ii) bound incidence
degeneracy. In (iii), every incidence subgraph is planar and bipartite.
Euler's formula gives average degree less than $4$ in every nonempty
such subgraph, so the incidence graphs are $3$-degenerate.
For (iv), let $F$ be the excluded topological minor and set $t=|V(F)|$.
Bollob\'as and Thomason~\cite{BollobasThomason1998} proved that, for a
universal constant $c$, every graph of average degree at least $c\,t^2$
contains a subdivision of $K_t$, and hence of $F$.
Every nonempty subgraph of an incidence graph therefore has average
degree less than $ct^2$, giving bounded incidence degeneracy.
\end{proof}
\section{Implications for Counting and the Polymatroid Width Hierarchy}
\label{sec:discussion}

We now ask whether fixed-parameter tractability of deciding $A\to B$
implies fixed-parameter tractability of counting all homomorphisms from
$A$ to $B$. For hypergraph classes, assuming $\FPT\ne\wone$, this amounts
to asking whether bounded submodular width suffices for fixed-parameter
exact counting. This remains open in general
\cite{KhamisEtAl2019FAQAI,KoutrisEtAl2025FasterJoins,jaguar}.

Bounded fractional hypertree width, on the other hand, permits
polynomial-time exact counting by dynamic programming on a suitable tree
decomposition. Our main theorem therefore settles the question on classes
of bounded incidence degeneracy or bounded primal independence. We also
obtain quantitative bounds between the widths governing the known decision
and counting algorithms.

Throughout this section, counting means counting all homomorphisms,
or equivalently all answers to a full conjunctive query. Write
$\mathsf{\#Hom}(\cH)$ for this counting problem on patterns whose
hypergraphs lie in $\cH$, and $p\text{-}\mathsf{\#Hom}(\cH)$ for its
parameterisation by $\lVert H_A\rVert$.

The algorithmic gap appears in PANDA, which evaluates Boolean queries
in time $\widetilde O(N^{\subw(H)})$ for a fixed query hypergraph $H$
and database size $N$~\cite{KhamisNgoSuciu2025PANDA}.
PANDA uses different decompositions for different parts of the data.
The resulting sets of satisfying assignments can overlap, which is
harmless for decision but prevents us from simply summing their counts.
Whether exact counting can achieve the same exponent remains open
\cite[Section~8]{KoutrisEtAl2025FasterJoins}.

The \#PANDA algorithm~\cite{KhamisEtAl2019FAQAI} ensures that each satisfying assignment contributes
through exactly one decomposition. Its exponent is sharp-submodular
width $\sharpsubw$. This condition is obtained by restricting the
submodularity inequalities available to the algorithm. Let
$\mathcal S_H^\sharp$ consist of the nonnegative, monotone,
edge-dominated functions $b$ with $b(\emptyset)=0$ satisfying
\[
 b(X\cup Y)+b(X\cap Y)\le b(X)+b(Y)
 \quad\text{whenever }X\cap Y\subseteq e
 \text{ for some }e\in E(H).
\]
Then $\sharpsubw(H)=\operatorname{width}_{\mathcal S_H^\sharp}(H)$,
and \#PANDA evaluates scalar aggregates over a commutative semiring
in $\widetilde O(N^{\sharpsubw(H)})$ semiring operations for fixed
$H$~\cite{KhamisEtAl2019FAQAI}.

Requiring submodularity only for intersections contained in an edge
\emph{enlarges} the family of functions over which width is maximised.
Thus the stronger algorithmic requirement for counting is reflected
in a potentially larger width. The known inequalities are
\cite[Proposition~3.16]{KhamisEtAl2019FAQAI}
\[
 \subw(H)\le\sharpsubw(H)\le\fhw(H).
\]
These inequalities alone do not bound $\sharpsubw$ in terms of
$\subw$, since bounded submodular width can coexist with unbounded
fractional hypertree width~\cite{adw,Marx13}.
Our main theorem bounds $\fhw(H)$, and hence $\sharpsubw(H)$, in terms of
$\subw(H)$ and either incidence degeneracy or primal independence.

\begin{corollary}
\label{cor:stoc-sharp-submodular-width}
For every hypergraph $H$ with $s=\subw(H)$,
\[
 \begin{aligned}
 s\le\sharpsubw(H)\le\fhw(H)&\le C\mu(H)\,s\log\max\{2,s\},\\
 \fhw(H)&\le C\max\{1,\log\alpha(H)\}\,s\log\max\{2,s\}
 \end{aligned}
\]
for a universal constant $C$.
\end{corollary}

That is, on classes with bounded $\mu$ or $\alpha$, it also settles the corresponding tractability question. There bounded submodular width implies polynomial-time exact counting.

While this settles the question of counting under bounded submodular width in the discussed setting, it is not immediately informative as to how counting can be achieved algorithmically through the ideas underlying submodular width based decision algorithms. The key problem there is that the algorithms create many partial instances that can overlap in their answers. As the number of partial instances is high, inclusion--exclusion to obtain a correct count is infeasible. We believe that \Cref{cor:stoc-sharp-submodular-width} can be additionally valuable as guidance in this direction. Can the $\mu$ or $\alpha$ parameters potentially also be used to control the overlap in partial instances, and to show that these submodular width based decision algorithms can be naturally extended to counting in some settings?

\begin{corollary}
\label{cor:bounded-lambda-counting}
Let $\cH$ be a recursively enumerable hypergraph class of bounded
$\mu$ or bounded $\alpha$. Assuming $\FPT\ne \wone$, the following are equivalent.
\begin{enumerate}[label=(\roman*)]
 \item $\cH$ has bounded submodular width.
 \item $\pHom(\cH)\in\FPT$.
 \item $p\text{-}\mathsf{\#Hom}(\cH)\in\FPT$.
 \item $\mathsf{\#Hom}(\cH)\in\PTIME$.
\end{enumerate}
\end{corollary}

\begin{proof}
As in \Cref{thm:bounded-lambda-collapse}, (i)$\Leftrightarrow$(ii) under $\FPT\ne \wone$.
For (i)$\Rightarrow$(iv), \Cref{thm:intro-width} gives bounded
fractional hypertree width. A decomposition of bounded width can be
found in polynomial time~\cite{MarxFHW2010}, and dynamic programming
on it counts all homomorphisms in polynomial
time~\cite{PichlerSkritek2013,DurandMengel2015}.
The implication (iv)$\Rightarrow$(iii) is immediate, and
(iii)$\Rightarrow$(ii) follows by testing whether the count is positive.
\end{proof}

On classes with bounded $\mu$ or $\alpha$, the implications
(i)$\Rightarrow$(iv)$\Rightarrow$(iii)$\Rightarrow$(ii) hold unconditionally. The assumption $\FPT\ne \wone$ is used only for (ii)$\Rightarrow$(i), through \Cref{thm:wone}. The corollary applies to all the structural classes in \Cref{cor:sparse-incidence-classes}.

The same bound also controls the intermediate polymatroid widths.
Let $\mathcal N_H$ be the edge-dominated normal polymatroids, namely
nonnegative combinations of the functions
$h_S(X)=\mathbf 1_{X\cap S\ne\emptyset}$ for
$\emptyset\ne S\subseteq V(H)$~\cite{KhamisEtAl2020BagContainment}.
Let $\mathcal E_H$ be the edge-dominated entropy functions of finite
discrete random variables, and let $\overline{\mathcal E}_H$ consist
of the edge-dominated limits of entropy functions. We have
\[
 \mathcal M_H\subseteq\mathcal N_H\subseteq\mathcal E_H
 \subseteq\overline{\mathcal E}_H\subseteq\mathcal S_H
 \subseteq\mathcal S_H^\sharp.
\]
Each of these function families defines a width parameter through
$\operatorname{width}_{\mathcal F}(H)$, just as modular and submodular
functions define adaptive and submodular width.
The intermediate cones are not merely formal interpolations. Abo Khamis,
Nakos, Olteanu, and Suciu~\cite{KhamisEtAl2024LpBounds} use normal, almost-entropic, and
polymatroid optimisation to obtain join-size bounds from $\ell_p$-norms
of degree sequences. These choices expose a trade-off between realisable
worst-case instances, asymptotic tightness, and efficient LP computation.
For simple degree sequences all three bounds coincide, with normal
relations witnessing tightness. Almost-entropic optimisation also gives
the entropic width introduced with PANDA~\cite{pandaexpress,KhamisNgoSuciu2026PANDAExpress}, whose
attainment as a query-evaluation exponent remains open~\cite{jaguar}.
Normal polymatroids also underlie decidable cases of bag query
containment~\cite{KhamisEtAl2020BagContainment} and the computation of
polymatroid bounds under simple degree
constraints~\cite{ImEtAl2025SimpleDegree}.
\Cref{thm:intro-width} bounds how far the decomposition widths obtained from
these function families can lie apart.

\begin{corollary}
\label{cor:width-hierarchy-band}
For every hypergraph $H$ with $a=\adw(H) \ge 2$,
\begin{equation}
\begin{aligned}
 a
 &\le\operatorname{width}_{\mathcal N_H}(H)
 \le\operatorname{width}_{\mathcal E_H}(H)
 \le\operatorname{width}_{\overline{\mathcal E}_H}(H)\\
 &\le\subw(H)\le\sharpsubw(H)\le\fhw(H)
 \le C\mu(H)\, a\log a,\\
 \fhw(H)&\le C\max\{1,\log\alpha(H)\}\, a\log a.
\end{aligned}
\label{eq:width-hierarchy-band}
\end{equation}
\end{corollary}

The proof is immediate from the inclusions and \Cref{thm:intro-width}.
The previously known bound $\subw(H)=O(a^4)$~\cite{Marx13} already shows
that the widths up to submodular width are simultaneously bounded on
arbitrary hypergraph classes. Under bounded incidence degeneracy or
bounded primal independence, our theorem extends this equivalence to
sharp-submodular width and fractional hypertree width. Moreover, all
these widths lie within an $O(\log a)$ factor of adaptive width.

\section{\wone-hardness for Homomorphism Problems with Unbounded Submodular Width}
\label{sec:ethsave}

In this section we present our proof that submodular width characterises fixed-parameter tractability of $\pHom$ also under the assumption that $\FPT \ne \wone$. For convenience we recall the statement from the introduction.

\wonehardnesstheorem*

We begin by recalling some necessary additional definitions and notation.
All connected sets, cliques, and paths in this section refer to the primal graph $P(H)$ defined in \Cref{sec:preliminaries}.
Two nonempty sets $U,W\subseteq V(H)$ \emph{touch} if they intersect or some hyperedge meets both. A \emph{bramble} is a family of nonempty connected vertex sets that pairwise touch.

Following Fan, Koutris, and Zhao~\cite{stocCliqueEmbedding2023}, a \emph{clique embedding} $\psi:K_m\to H$ assigns a nonempty connected
set $\psi(i)\subseteq V(H)$ to each $i\in[m]$, with $\psi(i)$ and $\psi(j)$
touching whenever $i\ne j$.  Its \emph{weak edge depth}
and the \emph{clique embedding power} of $H$ are
\begin{equation}\label{eq:embedding-power}
 d(\psi)=\max_{e\in E(H)}|\{i\in[m]\mid \psi(i)\cap e\ne\emptyset\}|,
 \qquad
 \emb(H)=\sup_{\substack{m\ge3\\\psi:K_m\to H}}\frac{m}{d(\psi)}.
\end{equation}

The images of a clique embedding form a bramble, with the subtlety that multiple clique vertices may have the same image. We can express this equivalently in terms of \emph{fractional packings} of a bramble $\mathcal B$. Such a packing assigns weights $w(A)\ge0$ to its
members such that
\[
 \sum_{\substack{A\in\mathcal B\\A\cap e\ne\emptyset}}w(A)\le1
 \qquad \forall e\in E(H).
\]
Thus each hyperedge meets bramble members of total weight at most one.
The value of the packing is $\sum_{A\in\mathcal B}w(A)$.

The maximum packing value, over all brambles in $H$, equals $\emb(H)$.
To see one direction, take an embedding $\psi:K_m\to H$ of depth $d$.
Give each distinct image $A$ weight equal to its multiplicity divided
by $d$. The depth bound makes this a feasible packing, of value $m/d$.
For the converse, take a rational packing of positive value $W$ and
choose a common denominator $q$ large enough that $qW\ge3$.
Use each member $A$ as an image
$qw(A)$ times. These images form a clique embedding with $qW$ clique
vertices and depth at most $q$, so its ratio is at least $W$.

We will use this clique embedding notion in the reduction but we find the packing formulation to be the more natural perspective in the structural lemma that connects the parameter to submodular width.

\subsection{Proof overview}
We reuse much of the machinery of \cite{Marx13} unchanged and only replace the final step that relied on ETH with an alternative approach.

We first show that unbounded clique embedding power suffices for hardness, by a direct reduction from the problem of finding a $k$-clique in a graph parameterised by $k$.

\begin{lemma}\label{lem:hard}
If a recursively enumerable class $\cH$ has unbounded clique
embedding power, then $\pHom(\cH)$ is $\wone$-hard under \FPT{}
many-one reductions.
\end{lemma}

For this reduction, we choose a hypergraph with $\emb(H)\ge2\binom{k}{2}$.
We show that this suffices to encode a $k$-clique instance on $N$ vertices.
Each target relation has at most $g(k)N^4$ rows, where $g$ is computable.

It then remains to obtain unbounded clique embedding power from
unbounded submodular width. For this we  build directly on one of the main combinatorial results of Marx's original proof which states, roughly, that high submodular width implies a large number of cliques that are highly connected among each other. We prove an additional observation on the structure of these cliques which will ultimately allow us to show that this implies high embedding power.

\begin{lemma}\label{lem:structural}
Every class of hypergraphs of unbounded submodular width has unbounded
clique embedding power.
\end{lemma}

We note that Fan, Koutris, and Zhao~\cite{stocCliqueEmbedding2023} already established the opposite direction, that unbounded clique embedding power implies unbounded submodular width, and the two notions are therefore equivalent in terms of boundedness over a class. In the rest of this section we only need the direction of Lemma~\ref{lem:structural}.

The two lemmas together thus replace the hardness direction of Theorem~\ref{thm:wone}, previously conditional on ETH, with \wone-hardness. The other direction, that bounded submodular width implies an \FPT{} algorithm, is not conditional on ETH.

\subsection{From clique embeddings to an \FPT{} reduction}
\label{sec:hardness}

Given a graph $G$ on $N$ vertices, we construct relational structures
$A_H$ and $B_H$ whose homomorphisms correspond to choices of $k$ pairwise
adjacent vertices (i.e., a clique) of $G$.

The reduction follows the usual embedding-based encoding of Clique into
the homomorphism problem, using the same local-assignment construction already used to show hardness under ETH~\cite[Lemma 7.5 and Remark 7.6]{Marx13}.
However, encoding graph vertices directly as domain values would give relation
tables of size $N^{O(d)}$, where $d$ is the embedding depth. Since $d$
depends on $k$, this does not give an \FPT{} reduction. Note that under ETH it was sufficient to have a function of $d$ in the exponent as long as it grew slowly enough.

The key trick to avoid this is to instead encode
each graph vertex using $s=\lceil d/2\rceil$ digits over an alphabet of
size $O(N^{1/s})$, i.e., we encode the vertices based on the clique embedding power. Our construction ensures that each relation involves
at most $2d$ digit coordinates, so its table has size at most
$f(k)N^{2d/s}\le f(k)N^4$.

\begin{proof}[Proof of Lemma~\ref{lem:hard}]
We reduce from the clique problem. Let $G$ be a simple graph with explicitly
listed vertex set $\{0,\ldots,N-1\}$. We may assume $k\ge3$ and $N\ge1$, since
the remaining cases are decidable directly. The indices $i\in[k]$
label the positions of the desired clique.

Let $r=\binom{k}{2}$. Using only $k$, find a hypergraph $H\in\cH$
and an optimal clique embedding $\psi:K_m\to H$ of weak depth $d$
with $m/d\ge2r$. To do so, enumerate $\cH$, computing an optimal
clique embedding for each hypergraph, until $\emb(H)\ge2r$. Unbounded clique
embedding power guarantees termination. We have that $H$, the embedding,
and hence $m$ and $d$ depend only on $k$. Also, note that $m$ is the size of the
embedded clique for the optimal clique embedding, whereas $k$ is the size of the clique sought in $G$.
\smallskip\textit{Construction of the homomorphism instance.}
Set $s=\lceil d/2\rceil$ and let
\(
 \mathcal R:=\{\!\{\psi(t)\mid t\in[m]\}\!\}.
\)
Process the pairs in $\binom{[k]}{2}$ one at a time, in any fixed order.
For the current pair $a$, use the current multiset $\mathcal R$ to choose
\[
 e_a\in\operatorname*{arg\,max}_{e\in E(H)}
 \bigl|\{\!\{A\in\mathcal R\mid A\cap e\ne\emptyset\}\!\}\bigr|,
\]
select $s$ members $A_{a,1},\ldots,A_{a,s}$ meeting $e_a$, and update
\(
 \mathcal R\leftarrow
 \{\!\{A\in\mathcal R\mid A\cap e_a=\emptyset\}\!\}
\), i.e., keep only those images of $\psi$ that do not intersect $e_a$.

To see that these choices are possible, each step removes at most $d$
members by definition of weak edge depth. Before any of the $r$ steps, $m'=|\mathcal R|$ therefore satisfies
\(
 m'\ge m-(r-1)d>m/2.
\)
The members of $\mathcal R$ give an embedding of $K_{m'}$ into $H$. If its depth $d'$
were at most $d/2$, then $m'/d'>m/d$, contradicting the optimality
of $\psi$. Thus $d'>d/2$, so $e_a$ meets at least $s$ members of
$\mathcal R$. Additionally observe that the  hyperedges $e_a$ are distinct, since after
its selection each edge meets no remaining member.

That is, we have gathered $\binom{k}{2}$ many distinct hyperedges $e_a$ that each meet at least $s$ distinct image occurrences of $\psi$, whose vertex sets may coincide. We will use this by encoding $a=\{i,j\}$ such that  the image $A_{a,q}$ will carry the $q$th coordinates of $k$-clique positions $i$ and $j$.

Towards this, let $B=\max\{2,\lceil N^{1/s}\rceil\}$ and let $D=\{0,\ldots,B-1\}$. We encode each vertex $v\in\{0,\ldots,N-1\}$ by the tuple $\operatorname{enc}(v)\in D^s$ of base-$B$ digits of $v$, padded
with leading zeros to length $s$. Since $B^s\ge N$, every vertex has such a unique representation.
We label coordinates as $P=[k]\times[s]$, with $(i,q)$ representing
the $q$th coordinate of the encoding of the $i$th selected vertex.
We store coordinate $(i,q)$ at every pattern vertex in the set
\[
 \Phi_{i,q}=\bigcup_{\substack{a\in\binom{[k]}{2}\\i\in a}}A_{a,q}.
\]
We call $\Phi_{i,q}$ the \emph{support} of coordinate $(i,q)$.
Each support is nonempty and connected, since its constituent sets
are nonempty, connected, and pairwise touching.
The relations below require copies of a coordinate stored at vertices
of a common hyperedge to agree. Connectedness then forces a single value
for $(i,q)$ throughout $\Phi_{i,q}$, so all checks involving position $i$
use the same $q$th digit.

Each pattern vertex stores values for the coordinates whose supports
contain it, using Marx's local-assignment
construction~\cite[Lemma 7.5 and Remark 7.6]{Marx13}. Write
\[
 P_v=\{(i,q)\in P:v\in\Phi_{i,q}\},\qquad
 U_e=\bigcup_{v\in e}P_v.
\]
Thus $P_v$ lists the coordinates stored at $v$, and $U_e$ lists all
coordinates stored at one or more vertices of $e$. We construct structures
$A_H$ and $B_H$ with $V(A_H)=V(H)$ and $V(B_H)=\widehat D$, where
\[
 \widehat D=\{(v,\alpha_v)\mid v\in V(H),\ \alpha_v:P_v\to D\}.
\]
Thus a value in $\widehat D$ specifies an assignment of digits from $D$
to the coordinates in $P_v$, together with the tag $v$ identifying
the pattern vertex it belongs to.

For each $e\in E(H)$, use a distinct relation symbol $R_e$ of arity
$|e|$ and fix an ordering $v_1,\ldots,v_\ell$ of the vertices of $e$.
Set $R_e^{A_H}=\{(v_1,\ldots,v_\ell)\}$. Thus $H_{A_H}=H$.
Define $R_e^{B_H}\subseteq\widehat D^\ell$ as follows.
The ordered vertices are the columns of $R_e^{B_H}$.
Enumerate all assignments $\alpha:U_e\to D$.
For each $\alpha$, construct the candidate tuple $(c_1,\ldots,c_\ell)$, where
\[
 c_j=(v_j,\alpha|_{P_{v_j}})\qquad\text{for }j=1,\ldots,\ell.
\]
Thus the entry in column $v_j$ consists of the tag $v_j$ and the digits
assigned by $\alpha$ to the coordinates in $P_{v_j}$.
All local assignments in this tuple agree on any coordinate they share.
At edges that are not checking hyperedges, keep every candidate row.
At a checking edge $e_a$ for $a=\{i,j\}$ with $i<j$, read the tuples
$(\alpha(i,q))_{q\in[s]}$ and $(\alpha(j,q))_{q\in[s]}$ as base-$B$
integers $x$ and $y$. Keep the row exactly when $x,y<N$ and
$\{x,y\}\in E(G)$. Thus both tuples must encode vertices of $G$,
and those vertices must be adjacent. Let $R_e^{B_H}$ consist of the rows
kept in this way.

\smallskip\textit{Correctness.}
That is, the specially picked edges $e_a$ encode adjacency of the pair $a\in \binom{[k]}{2}$, while all other edges are constructed in such a way that they do not create any further constraints. The existence of a homomorphism $A_H\to B_H$ is then intuitively equivalent to the existence of $k$ pairwise adjacent vertices in $G$. We move on to proving this formally.

Suppose $u_1,\ldots,u_k$ form a clique in $G$. Define
$h:V(A_H)\to V(B_H)$ by $h(v)=(v,\alpha_v)$, where
\[
 \alpha_v(i,q)=\operatorname{enc}(u_i)_q
 \qquad\text{for every }(i,q)\in P_v.
\]
All copies of a coordinate receive the same digit, so the tuple assigned
to each edge is among the candidates used to define its relation.
At a checking edge $e_{\{i,j\}}$, this tuple encodes $u_i$ and $u_j$,
which are adjacent in $G$, so the tuple is kept. All other edges keep
every candidate tuple. Thus $h$ preserves every relation and is a
homomorphism from $A_H$ to $B_H$.

Conversely, suppose $h:A_H\to B_H$ is a homomorphism. Each element
$v\in V(A_H)$ has $h(v)=(v,\alpha_v)$ with $\alpha_v:P_v\to D$.
Indeed, every element occurs in a tuple of $A_H$, and every tuple in
the corresponding relation of $B_H$ has tag $v$ in column $v$.

For each coordinate $(i,q)$, the relations enforce agreement of its
copies within each hyperedge. Since its support $\Phi_{i,q}$ is connected,
all copies therefore have a common value $b_{i,q}$; that is,
\[
 \alpha_v(i,q)=b_{i,q}\qquad\text{for every }v\in\Phi_{i,q}.
\]

For each $i\in[k]$, let $u_i$ be the integer with base-$B$ digits
$(b_{i,1},\ldots,b_{i,s})$. For every pair $\{i,j\}$, the relation
for $e_{\{i,j\}}$ tests exactly these integers $u_i$ and $u_j$.
Since $h$ preserves this relation, $u_i,u_j<N$ and
$\{u_i,u_j\}\in E(G)$. Thus $u_1,\ldots,u_k$ are pairwise adjacent
vertices of $G$. They are distinct because $G$ has no loops, and hence
form a $k$-clique.

\smallskip
\textit{Efficiency of the reduction.}
Constructing $R_e^{B_H}$ requires considering $B^{|U_e|}$ assignments
$\alpha:U_e\to D$. Let us therefore observe an upper  bound on the number $|U_e|$.
By the depth of $\psi$, $e$ meets at most $d$ of the selected images
$A_{a,q}$. Each such image $A_{\{i,j\},q}$ contributes the two
coordinate labels $(i,q)$ and $(j,q)$, and these account for all of
$U_e$. Thus
\(
 |U_e|\le2d\le4s.
\)
The number of candidate tuples for $R_e^{B_H}$ is consequently at most
\begin{equation}\label{eq:tables}
 B^{|U_e|}\le B^{4s}\le(2N^{1/s})^{4s}=2^{4s}N^4.
\end{equation}
Here $B\le2N^{1/s}$ also holds when $N=1$. The tuple encoding
has thus turned the depth-dependent number of coordinates into the
constant exponent $4$ on $N$.

The hypergraph, supports, and $s$ depend only on $k$. Constructing the relations
requires enumerating at most $B^{4s}$ maps per edge and testing encoded
adjacency in polynomial time per map. Since $H$ has no isolated vertices,
each $P_v$ is contained in some $U_e$, so $V(B_H)$ has at most
$|V(H)|B^{4s}$ elements. Thus the construction takes
$f(k)|G|^{O(1)}$ time for a computable $f$, and the constructed structure
$A_H$ depends only on $k$.
\end{proof}

\subsection{From submodular width to clique embeddings}
\label{sec:structural}

We move on to proving that high submodular width implies high clique embedding power. We start from a core result of \cite{Marx13} that shows that high submodular width implies the existence of disjoint cliques whose union satisfies a strong connectivity condition. We call these \emph{terminal cliques}, since they will serve as endpoints for the paths constructed below. Our plan will be to arrange these terminal cliques in a grid such that it gives a bramble of connected sets meeting a row and a column. Via the equivalence between bramble packings and clique embedding power described above this then proves the desired result.

We first recall some important notions from~\cite{Marx13}.
For a nonnegative vertex weighting $\mu$, write $\mu(X)=\sum_{v\in X}\mu(v)$.
A fractional $(X,Y)$-separator is a
nonnegative edge weighting $s$ such that every $X$--$Y$ path $P$ satisfies
$\sum_{e:e\cap V(P)\ne\emptyset}s(e)\ge1$.
These fractional $(X,Y)$-separators assign weights to hyperedges and differ from the fractional balanced separators in \Cref{sec:separator-definitions}.
A set $W$ is \emph{$(\mu,\lambda)$-connected} if, for all disjoint
$X,Y\subseteq W$, every fractional $(X,Y)$-separator $s$ has total weight
$\sum_{e\in E(H)}s(e)$ at least
$\lambda\min\{\mu(X),\mu(Y)\}$.
We use the following version of results from~\cite{Marx13}.

\begin{lemma}[Consequence of Marx~{\cite[Theorem 5.1 and Lemma 6.5]{Marx13}}]
\label{lem:marx-cliques}
There are absolute constants $\beta,c_T,w_0>0$ such that every hypergraph
$H$ with $w=\subw(H)\ge w_0$ admits a nonnegative vertex weighting
$\mu$ and disjoint nonempty cliques $K_1,\ldots,K_k$ such that
\[
 k\ge c_T\sqrt w,\qquad \mu(K_i)\ge1/2\quad(i\in[k]),
\]
and $\bigcup_{i=1}^kK_i$ is $(\mu,\beta)$-connected.
\end{lemma}

\begin{proof}
Fix a sufficiently small constant $0<\lambda<1/16$ as required by
Marx's Theorem 5.1. His parameter $\operatorname{con}_\lambda(H)$ is the
maximum $\nu(U)$ over $(\nu,\lambda)$-connected sets $U$ and nonnegative
vertex weightings $\nu$ satisfying $\nu(e)\le1$ for every hyperedge.
For $w>3$, we have $\frac32(w/3+1)=w/2+3/2<w$. The definition of
submodular width therefore supplies an edge-dominated polymatroid $b$
with $\operatorname{width}_{\{b\}}(H)>\frac32(w/3+1)$, so Theorem 5.1 gives
$\operatorname{con}_\lambda(H)\ge w/3$.
Marx's Lemma 6.5 then supplies the stated cliques with
$k=\Omega(\lambda\sqrt{\operatorname{con}_\lambda(H)})=\Omega(\sqrt w)$
and connectivity constant $\beta=\lambda/6$. Choose $c_T>0$ and $w_0$
so that the bound on $k$ holds for all $w\ge w_0$.
\end{proof}

\subsubsection{Cheap hubs}\label{sec:metric}

We want to connect the cliques in a row and a column of the grid without
making the connecting paths too expensive. One way is to connect them
through a common clique, which serves as a hub. This hub need not be close
to every clique. It is enough that its total distance to the cliques is
small, since we can then choose a row and a column that are cheap to
connect. The reason to expect such a hub is the strong connectivity of
the cliques, which prevents large groups of them from being far apart.
We first establish the metric statement underlying this idea.

A \emph{pseudometric} on $[k]$ is a function
$d:[k]\times[k]\to\mathbb R_{\ge0}$ such that, for all $i,j,l\in[k]$,
\[
 d(i,i)=0,\qquad d(i,j)=d(j,i),\qquad
 d(i,j)\le d(i,l)+d(l,j).
\]
Distinct points may have distance zero, as needed below when terminal
cliques are joined by paths of zero cost.

For nonempty sets $A,B\subseteq[k]$, write
$d(A,B)=\min_{i\in A,j\in B}d(i,j)$. Connectedness will bound
the separation between groups of terminals. The following lemma says
that if no two large groups are far apart, there is a point whose total
distance to all the others is small. We call such a point a \emph{hub}.

\begin{lemma}\label{lem:hub}
There is an absolute constant $C_H>0$ such that, for every pseudometric
$d$ on $[k]$ with $k\ge2$, some $z\in[k]$ satisfies
\begin{equation}\label{eq:hub}
 \sum_{u\in[k]}d(z,u)\le C_H\log k
 \max_{\substack{\emptyset\ne A,B\subseteq[k]\\A\cap B=\emptyset}}
       \min\{|A|,|B|\}\,d(A,B).
\end{equation}
\end{lemma}

\begin{proof}
Let us write $M$ for the maximum on the right.
If $M=0$, taking $A=\{u\}$ and $B=\{v\}$ in its definition gives
$d(u,v)\le M=0$ for every $u\ne v$. Hence every distance is zero,
and any point satisfies the claim. For the rest of the proof, assume
$M>0$ and put $r=8M/k$ and $L=\lceil\log_2 k\rceil$.
By the definition of $M$, disjoint sets $A,B\subseteq[k]$ satisfy
\begin{equation}\label{eq:large-set-separation}
 |A|,|B|\ge k/8\quad\Longrightarrow\quad d(A,B)\le r.
\end{equation}

We first find a ball of radius $Lr$ containing at least $k/8$ points.
Enlarging its radius to $R=(L+1)r$ will include more than $7k/8$ points.
These points contribute at most $kR=O(M\log k)$ to the total distance
from the centre. For the remaining points, the definition of $M$ bounds
the $j$th largest distance by $R+M/j$. Summing these bounds gives the
required total of $O(M\log k)$.

For $U\subseteq[k]$, write
$B_U(x,t)=\{u\in U:d(x,u)\le t\}$ for a closed ball restricted to $U$.

\smallskip
\textit{Finding a good ball.}
Start with $U=[k]$. While $U\ne\emptyset$, choose $x\in U$ and let
\[
 j=\min\bigl\{t\in\{0,\ldots,L-1\}:
 |B_U(x,(t+1)r)|\le2|B_U(x,tr)|\bigr\}.
\]
Record $B_U(x,jr)$ as a \emph{core} and remove $B_U(x,(j+1)r)$ from $U$.
Repeat until $U$ is empty.

Such a $j$ exists because the initial ball contains $x$, and $L$
enlargements that each more than doubled its size would produce more
than $2^L\ge k$ points.

This construction has two properties. First, the removed balls
partition $[k]$, and each core contains at least half the points of its
corresponding removed ball. Thus the disjoint cores contain at least $k/2$ points in
total. Second, different cores are more than $r$ apart. Indeed, if $a$
belongs to a core centred at $x$ with radius $jr$, and $b$ belongs
to a later core, then $d(x,a)\le jr$ and $d(x,b)>(j+1)r$. Therefore
$d(a,b)\ge d(x,b)-d(x,a)>r$.

Suppose every core has fewer than $k/8$ points. Add cores to a union $A$
until $|A|\ge k/8$. Before the last addition, $A$ had fewer than $k/8$
points, and the added core also has fewer than $k/8$ points. Thus
$|A|<k/4$. Since the cores together contain at least $k/2$ points,
the union $B$ of the remaining cores satisfies
\(
 |B|\ge k/2-|A|>k/4.
\)
Both unions therefore contain at least $k/8$ points. But distinct cores
are more than $r$ apart, so $d(A,B)>r$, contradicting
\eqref{eq:large-set-separation}. Hence some core contains at least $k/8$
points. Let $z$ be its centre.
The ball $D_0=B_{[k]}(z,Lr)$ in the full space contains this core.

Set $R=(L+1)r$ and $D=B_{[k]}(z,R)$. For $x\notin D$ and $y\in D_0$,
\[
 d(x,y)\ge d(x,z)-d(y,z)>(L+1)r-Lr=r.
\]
Thus $D_0$ and the complement of $D$ are more than $r$ apart.
Since $D_0$ contains at least $k/8$ points,
\eqref{eq:large-set-separation} implies that the complement contains
fewer than $k/8$ points. Thus $|D|>7k/8$.

\smallskip
\textit{Controlling distances from the centre.}
The points in $D$ contribute at most $R$ each to the desired sum.
Write the distances $d(z,u)$, for all $u\in[k]$, in decreasing order
\[
 \rho_1\ge\cdots\ge\rho_k=0.
\]
We first bound those distances that exceed $R$.
If $\rho_j>R$, let $F$ be the set of the $j$ farthest points. All of
them lie outside $D$, so $|F|=j<k/8<|D|$. Moreover, the triangle
inequality gives $d(D,F)\ge\rho_j-R$. Applying the definition of $M$
to $D$ and $F$ yields
\[
 j(\rho_j-R)\le M,\qquad\text{and hence}\qquad
 \rho_j\le R+\frac{M}{j}.
\]
The last bound also holds when $\rho_j\le R$. We may therefore sum it
over all $j$. Writing $H_k=\sum_{j=1}^k1/j$, we obtain
\[
 \sum_{u\in[k]}d(z,u)
 \le kR+MH_k
 =8M(L+1)+MH_k
 =O(M\log k).
\]
\end{proof}

\subsubsection{From terminal cliques to fractional bramble packings}\label{sec:lift}

We arrange the terminal cliques in a square grid. Each member of our
bramble will be a connected set meeting all cliques in one row and one
column. Any two such sets touch because their crosses share a terminal
clique. The useful feature is that an $n\times n$ grid has $n^2$ crosses,
while each terminal belongs to at most $2n$ of them. For any choice of
hyperedge costs, the hub lemma supplies paths connecting the terminals
with small total cost. Using these paths to connect each cross counts
every path at most $2n$ times. Their cost is therefore spread across
$n^2$ crosses, giving a much smaller average cost per cross. The proof
turns this saving into a lower bound on the fractional packing value.

\begin{proposition}\label{prop:terminal-embedding}
Let $\mu$ be a nonnegative vertex weighting of $H$, and let
$K_1,\ldots,K_k$ be disjoint cliques, with $k\ge8$.
Suppose their union is $(\mu,\lambda)$-connected for some $\lambda>0$,
and $\mu(K_i)\ge1/2$ for every $i$.
Then
\begin{equation}\label{eq:lift}
 \emb(H)\ge\frac{\lambda\sqrt{k}}{8C_H\log k}.
\end{equation}
\end{proposition}

\begin{proof}
All terminal cliques lie in one connected component. Otherwise, the zero
edge weighting would separate two terminal cliques of positive mass.
Put $n=\lfloor\sqrt{k}\rfloor\ge2$ and $m=n^2$. Pick any $m$ of the $k$ cliques and index them by $T=[n]^2$ and note that their union is $(\mu,\lambda)$-connected. We will implicitly treat $[n]^2$ as a grid, with rows being those pairs that agree in the first coordinate, and columns those that agree in the second.

For each $a=(i,j)\in[n]^2$, let $\mathcal B_a$ consist of all
inclusion-minimal connected vertex sets meeting every terminal clique
in row $i$ and every terminal clique in column $j$. We can then define
\(
 \mathcal B=\bigcup_{a\in[n]^2}\mathcal B_a.
\)
Each $\mathcal B_a$ is nonempty because all terminal cliques lie in one
connected component and $H$ is finite. For $A\in\mathcal B_{(i,j)}$ and
$A'\in\mathcal B_{(i',j')}$, both sets meet $K_{(i,j')}$, so they touch
because $K_{(i,j')}$ is a clique. Hence $\mathcal B$ is a bramble. It will be our goal to show that this bramble admits a large fractional packing.

\smallskip
For a nonnegative edge weighting $y:E(H)\to\mathbb R_{\ge0}$, define
\[
 c_y(S)=\sum_{e:e\cap S\ne\emptyset}y(e).
\]
By LP duality, the maximum fractional packing value of $\mathcal B$ equals
\begin{equation}\label{eq:bramble-dual}
 \min\left\{\sum_{e\in E(H)}y(e):
 y\ge0,\quad c_y(A)\ge1\text{ for every }A\in\mathcal B\right\}.
\end{equation}
Fix any weighting satisfying these constraints, and put $p=\sum_e y(e)$.
We will show that
\[
 p\ge\frac{\lambda n}{4C_H\log m},
\]
which gives the same lower bound on the fractional packing value.

\smallskip
\textit{Bounding terminal distances.}
For distinct $u,v\in T$, let $d(u,v)$ be the minimum cost $c_y(V(P))$
of a $K_u$--$K_v$ path, and set $d(u,u)=0$. These distances form a
pseudometric. For distinct $u,v,w$, minimum-cost paths from $K_u$ to
$K_w$ and from $K_w$ to $K_v$ can be joined inside $K_w$ without adding
vertices. Their union therefore contains a $K_u$--$K_v$ path of cost
at most $d(u,w)+d(w,v)$, proving the triangle inequality.
For disjoint nonempty $I,J\subseteq T$, connectedness implies
\[
 \min\{|I|,|J|\}\,d(I,J)\le\frac{2p}{\lambda}.
\]
Otherwise, putting $\delta=d(I,J)>0$, the weights $y(e)/\delta$
would separate the two groups with total weight
\[
 \frac{p}{\delta}<\frac{\lambda}{2}\min\{|I|,|J|\},
\]
contradicting $(\mu,\lambda)$-connectedness, since their masses are
at least $|I|/2$ and $|J|/2$ (recall the terminal cliques are disjoint and each has $\mu(K_i) \ge 1/2$ by assumption).

Lemma~\ref{lem:hub} therefore gives a hub $z\in T$. For every $u\ne z$,
choose a minimum-cost $K_z$--$K_u$ path $P_u$. Their total cost satisfies
\[
 \sum_{u\ne z}c_y(V(P_u))=\sum_{u\in T}d(z,u)
 \le\frac{2C_H\log m}{\lambda}p.
\]

\smallskip
We move on to establishing the final packing bound as described above.
For each $a=(i,j)\in[n]^2$, set
\[
 S_a=\bigcup_{\substack{(r,c)\in[n]^2\setminus\{z\}\\
                        r=i\ \text{or}\ c=j}}V(P_{(r,c)}).
\]
Thus $S_a$ contains the vertices of the chosen paths to cliques in row
$i$ or column $j$. The set $S_a$ is connected and meets all the required
terminal cliques. Choose an inclusion-minimal connected subset
$A_a\subseteq S_a$ that still meets every terminal clique in row $i$
and column $j$. Such a subset exists because $S_a$ has these properties
and is finite. Then $A_a\in\mathcal B_a$.

We first bound the cost of this one member. Since $A_a\subseteq S_a$,
its cost is at most $c_y(S_a)$. Every hyperedge meeting $S_a$ meets at
least one of its paths, so summing their costs can only increase the
total. Finally, each $P_u$ was chosen with cost $c_y(V(P_u))=d(z,u)$.
Therefore
\[
 \begin{aligned}
 c_y(A_a)&\le c_y(S_a)
 &\le\sum_{\substack{u=(r,c)\in T\setminus\{z\}\\
                     r=i\ \text{or}\ c=j}}c_y(V(P_u))
 &=\sum_{\substack{u=(r,c)\in T\setminus\{z\}\\
                   r=i\ \text{or}\ c=j}}d(z,u).
 \end{aligned}
\]

Now sum this inequality over all pairs $a=(i,j)\in[n]^2$ to obtain
\[
 \sum_{a\in[n]^2}c_y(A_a)
 \le\sum_{(i,j)\in[n]^2}
       \sum_{\substack{u=(r,c)\in T\setminus\{z\}\\
                       r=i\ \text{or}\ c=j}}d(z,u).
\]
For each fixed $u=(r,c)\ne z$, the term $d(z,u)$ appears for the $n$
pairs with $i=r$ and the $n$ pairs with $j=c$, with $(r,c)$ counted
in both groups. Thus it appears $2n-1$ times. Since $A_a\in\mathcal B$,
the dual constraints in~\eqref{eq:bramble-dual} give $c_y(A_a)\ge1$
for every $a$. Reordering the double sum therefore yields
\[
 n^2\le\sum_{a\in[n]^2}c_y(A_a)
 \le2n\sum_{u\ne z}d(z,u)
 \le\frac{4C_Hn\log m}{\lambda}p.
\]
Rearranging and using $n\ge\sqrt{k}/2$ and $m\le k$ gives
\[
 p\ge\frac{\lambda n}{4C_H\log m}
 \ge\frac{\lambda\sqrt{k}}{8C_H\log k}.
\]
Since $y$ was arbitrary, LP duality gives this lower bound on the
fractional packing value of $\mathcal B$, and therefore on $\emb(H)$.
\end{proof}

\begin{proof}[Proof of Lemma~\ref{lem:structural}]
By Lemma~\ref{lem:marx-cliques}, unbounded submodular width supplies
arbitrarily many disjoint terminal cliques of $\mu$-mass at least $1/2$
whose union is $(\mu,\beta)$-connected, for a fixed constant $\beta>0$.
For $k\ge8$ such cliques,
Proposition~\ref{prop:terminal-embedding} gives
\[
 \emb(H)\ge\frac{\beta\sqrt{k}}{8C_H\log k}
 \longrightarrow\infty\qquad\text{as }k\to\infty.
\]
\end{proof}

\section{Further Details on Related Work}
\label{sec:related-work}
\paragraph{Previous work on $\adw$/$\fhw$ comparisons}
Prior work~\cite{LanzingerPODS22} characterises tractability for recursively
enumerable classes of hypergraphs of maximum degree two by bounded
generalised hypertree width, assuming $\FPT\ne \wone$.
It also yields that bounded fractional hypertree width and bounded
submodular width coincide on such classes. However, this consequence is
conditional on complexity assumptions and only qualitative, with no
explicit numerical bound relating the widths, since it follows through
the complexity characterisations.

In recent work, Lanzinger~\cite{Lanzinger2026Cuts} uses multicommodity
flows to obtain, for $\adw(H)\ge2$,
\[
 \ghw(H)=O\!\left(\adw(H)\log\adw(H)\right)
\]
on classes of bounded edge excess and on hypergraphs with private
intersections. The constants may depend on the edge-excess bound.
Although the main statements use submodular width, the modular
witnesses in~\cite[Section~6.1]{Lanzinger2026Cuts} give the same lower
bound for adaptive width. We write the equivalent upper bound on
$\ghw$ to make the comparison with \Cref{thm:intro-width} explicit.
Bounded edge excess here means that, in each hyperedge, the
sum of the amounts by which vertex degrees exceed two is bounded by a
fixed constant. A hypergraph has private intersections if every nonempty
intersection of two distinct hyperedges contains a vertex belonging to
exactly those two hyperedges.

\Cref{thm:intro-width} gives an unconditional quantitative comparison
for every bounded-degree class and, more generally, every class of
bounded incidence degeneracy. For $\adw(H)\ge2$,
\[
 \fhw(H)=O\!\left(\Delta(H)\,\adw(H)\log\adw(H)\right).
\]
Thus our result extends the degree-two and bounded-edge-excess
comparisons to every bounded-degree class. The two bounds control
different widths, but
\[
 \fhw(H)\le\ghw(H)\le\Delta(H)\fhw(H).
\]
Indeed, in a fractional cover, each covered vertex belongs to an edge
of weight at least $1/\Delta(H)$. Choosing all such edges gives an
integral cover of at most $\Delta(H)$ times the fractional cost.
Bounded edge excess implies bounded degree, whereas bounded degree
allows arbitrarily many degree-three vertices in one hyperedge.

The private intersection case remains complementary, since that condition permits incidence degeneracy and primal independence number to be simultaneously arbitrarily large. The flow and separator arguments therefore seem to apply under different structural
conditions. We discuss this further in
\Cref{sec:stoc-open-problems}.

\paragraph{Connections to fractional hypertree-width approximation}
The proof of \Cref{thm:intro-width} uses the separator results in~\cite{KorchemnaEtAl2024}. These results were developed for a polynomial-time approximation algorithm for fractional hypertree width. For a width bound
$k\ge2$, their guarantee includes a decomposition of width
$O(\mu(H)k\log k)$ and $O(\max\{1,\log\alpha(H)\}k\log k)$ whenever $\fhw(H)\le k$.
Their proof introduces the fractional balanced-separator relaxation and
the rounding theorem used in \Cref{sec:separators}. Their algorithm combines
this rounding theorem with a recursive separator construction to approximate
$\fhw$, using $\fhw(H)$ itself to bound the optimum separator cost. Our
minimax argument instead bounds that cost by $\adw(H)$ uniformly over all
induced subhypergraphs and hyperedge weightings, which turns their rounding
theorem into a direct comparison between the two widths. In particular, the
dependence on incidence degeneracy and primal independence comes
from their rounding analysis. We prove
$\fsep_{1/2}(H)\le\adw(H)$ by minimax and the balanced-bag property of
tree decompositions.
\section{Conclusion}
\label{sec:conclusion}

We proved that every hypergraph $H$ with $\adw(H)\ge2$ satisfies
\[
 \begin{aligned}
 \fhw(H)&=O\!\left(\mu(H)\adw(H)\log\adw(H)\right),\\
 \fhw(H)&=O\!\left(\max\{1,\log\alpha(H)\}\adw(H)\log\adw(H)\right).
 \end{aligned}
\]
Together with the known width and complexity characterisations, this shows
that on recursively enumerable hypergraph classes of bounded incidence
degeneracy or bounded primal independence,
bounded adaptive width, bounded submodular width, bounded fractional
hypertree width, fixed-parameter tractability of $\pHom$, and polynomial-time
solvability of $\Hom$ are equivalent under $\FPT\ne \wone$. The same characterisation
holds for recursively enumerable classes of structures with bounded
incidence degeneracy after taking cores. The width bound also gives
polynomial-time exact counting on hypergraph classes of bounded submodular
width and either bounded incidence degeneracy or bounded primal independence.
In general, it gives a quasipolynomial-time
algorithm for $\Hom$ whenever $\pHom$ is fixed-parameter tractable, assuming
$\FPT\ne \wone$.

\subsection{Open Problems}
\label{sec:stoc-open-problems}

\paragraph{NP-intermediate homomorphism problems.}
Can one exhibit a naturally defined, polynomial-time recognisable
hypergraph class $\cH$ with
\[
 \sup_{H\in\cH}\subw(H)<\infty
 \qquad\text{and}\qquad
 \Hom(\cH)\notin\PTIME,
\]
where the second assertion is proved under a standard complexity
assumption? Note also that by our results, the bounded submodular width implies a quasipolynomial-time algorithm unconditionally. So if $\NP\nsubseteq\QP$, the problem $\Hom(\cH)$ would be $\NP$-intermediate. Ideally, one would identify a structural criterion implying $\Hom(\cH)\notin\PTIME$ that is satisfied by some naturally defined class of bounded submodular width. Such a criterion could turn bounded submodular width into a source of natural $\NP$-intermediate problems. The prospect of explaining intermediate complexity through the structure of the patterns makes this question particularly compelling.

\paragraph{Unifying the flow and separator proofs.}
Our separator argument and the multicommodity-flow argument
of~\cite{Lanzinger2026Cuts} give closely related comparisons under
incomparable structural conditions. For $\adw(H)\ge2$, they yield
\[
 \begin{aligned}
 \fhw(H)&=O(\mu(H)\adw(H)\log\adw(H)),\\
 \fhw(H)&=O(\max\{1,\log\alpha(H)\}\adw(H)\log\adw(H))
 \end{aligned}
\]
and, for hypergraphs with private intersections
\[
 \ghw(H)=O(\adw(H)\log\adw(H)).
\]
The private intersections condition permits incidence degeneracy and primal independence number to be simultaneously arbitrarily large and is not
preserved by vertex restriction. Indeed, every
hypergraph is an induced subhypergraph of one with private intersections,
obtained by adding, for each pair of intersecting edges, a new vertex
belonging to exactly those two edges.
 A common proof would likely instead have to use the private witnesses globally, or replace recursive restriction by a
decomposition principle compatible with both vertex separators and edge
cuts.

The logarithmic losses also enter at different points. In our proof they
originate from rounding a fractional balanced separator of cost $r$
to an integral separator.  The subsequent separator recursion loses only a
constant factor. In the flow argument the logarithm comes from the
$O(\log m)$ node-capacitated flow--cut gap for product demands on a
well-linked set of size $m$. It is unlikely that either approach can be strengthened to avoid the additional logarithmic factor.

No example is currently known in which either logarithmic loss is
necessary. This leaves open the possibility that
\[
 \fhw(H)=O(\mu(H)\adw(H)),
 \qquad
 \fhw(H)=O(\max\{1,\log\alpha(H)\}\adw(H))
\]
hold for all hypergraphs and that
\(
 \ghw(H)=O(\adw(H))
\)
holds for hypergraphs with private intersections. Proving any of these bounds
appears to require a more direct structural argument. Conversely, a
matching logarithmic separation would identify a genuine obstruction
shared by the two settings.

\paragraph{Further classes admitting separator rounding.}
The inequality $\fsep_{1/2}(H)\le\adw(H)$ holds for every hypergraph.
A broader problem is to identify structural conditions beyond bounded
incidence degeneracy or bounded primal independence that permit effective
rounding of these fractional separators.
If a class admits a nondecreasing
function $g$ and a fixed $\eta\in[1/2,1)$ such that every fractional
$(\gamma,1/2)$-balanced separator $x$ in every hypergraph
$J$ of the class can be rounded to a $(\gamma,\eta)$-balanced separator $S$ with
$\rho_J^*(S)\le g(\rho_J^*(x))$, then the argument using
\Cref{lem:recursive-decomposition} gives
\[
 \fhw(H)\le\frac{2g(\adw(H))}{1-\eta}+g(\adw(H)).
\]
Thus further rounding theorems would yield new bounds on fractional
hypertree width in terms of adaptive width and new classes on which \FPT{}
decision implies polynomial-time exact counting, by the argument of
\Cref{cor:bounded-lambda-counting}.

\ifaidisclosure
\section*{Full Declaration on the use of Generative AI}
We follow the format recommended by the AI+TCS Working Group at the Simons Institute for the Theory of Computing~\cite{SimonsInstituteAITCS2026}.
\begin{flushleft}
\renewcommand{\arraystretch}{1.3}
\begin{tabular}{@{}l@{\qquad}c@{}}
\textbf{Activity} & \textbf{AI used} \\
Asking the research question & [No] \\
Coming up with the approach & [No] \\
Proof development & [Yes] \\
Writing and exposition & [No] \\
Checking for bugs & [Yes] \\
Other supporting tasks & [Yes] \\
\end{tabular}
\end{flushleft}

\begin{sloppypar}
\paragraph*{Detailed statement}
Regarding proof development, various proofs were developed in co-work with OpenAI GPT-6 Astra. The process generally consists of starting from a statement and proof sketch and then iteratively building up the details. Once author and AI agree, the proofs are written out by the author and proofread again by AI.

Moreover, when presented with the general proof strategy, OpenAI GPT-6 Astra proposed the statement of \Cref{lem:hub}, although the original proof was based on Bourgain's embedding theorem. This was replaced with a more elementary argument for easier verifiability.

After completion of the first version of this manuscript, OpenAI GPT-6 Astra was used to formalise the proof of \Cref{thm:intro-width} in Lean. The Lean proof and blueprint are available under~\url{https://github.com/lnz/sparse-incidence-homomorphisms}. This includes also formalisations of the key \Cref{thm:rounding} and a version of \Cref{lem:recursive-decomposition}.
\end{sloppypar}
\fi

\bibliographystyle{alphaurl}
\newcommand{\etalchar}[1]{$^{#1}$}

\end{document}